\pdfoutput=1
\documentclass[acmsmall]{acmart}

\setcopyright{none}
\usepackage{booktabs}   %% For formal tables:
\usepackage{subcaption} %% For complex figures with subfigures/subcaptions
\usepackage{enumerate}

\usepackage{amsthm}
\usepackage{amsmath}
\usepackage{amssymb}

\usepackage{algorithm}
\usepackage{algpseudocode}
\usepackage{multirow}
\usepackage{url}
\usepackage{subcaption}
\usepackage{enumitem}
\usepackage{tabularx}
\usepackage[font=small,labelfont=bf,margin=\parindent,tableposition=top]{caption}

\algnewcommand{\algorithmicand}{\textbf{ and }}
\algnewcommand{\algorithmicor}{\textbf{ or }}
\algnewcommand{\OR}{\algorithmicor}
\algnewcommand{\AND}{\algorithmicand}

\newcommand{\etal}{\textit{et al}.}

\begin{document}

%% Title information
\title{Accelerating Concurrent Heap on GPUs}         %% [Short Title] is optional;
                                        %% when present, will be used in
                                        %% header instead of Full Title.
             %% \titlenote is optional;
                                        %% can be repeated if necessary;
                                        %% contents suppressed with 'anonymous'
%\subtitle{Subtitle}                     %% \subtitle is optional
%\subtitlenote{with subtitle note}       %% \subtitlenote is optional;
                                        %% can be repeated if necessary;
                                        %% contents suppressed with 'anonymous'

%% Author information
%% Contents and number of authors suppressed with 'anonymous'.
%% Each author should be introduced by \author, followed by
%% \authornote (optional), \orcid (optional), \affiliation, and
%% \email.
%% An author may have multiple affiliations and/or emails; repeat the
%% appropriate command.
%% Many elements are not rendered, but should be provided for metadata
%% extraction tools.

%% Author with single affiliation.
\author{Yanhao Chen}
\affiliation{
  \department{Computer Science}              %% \department is recommended
  \institution{Rutgers University}            %% \institution is required
  \country{USA}                    %% \country is recommended
}
\email{yc827@cs.rutgers.edu}          %% \email is recommended

\author{Fei Hua}
\affiliation{
  \department{Computer Science}              %% \department is recommended
  \institution{Rutgers University}            %% \institution is required
  \country{USA}                    %% \country is recommended
}
\email{huafei90@gmail.com}          %% \email is recommended

\author{Chaozhang Huang}
\affiliation{
  \department{Computer Science}              %% \department is recommended
  \institution{Rutgers University}            %% \institution is required
  \country{USA}                    %% \country is recommended
}
\email{ch616@scarletmail.rutgers.edu}          %% \email is recommended

\author{Jeremy Bierema}
\affiliation{
  \department{Computer Science}              %% \department is recommended
  \institution{Rutgers University}            %% \institution is required
  \country{USA}                    %% \country is recommended
}
\email{jeremy.bierema@rutgers.edu}          %% \email is recommended

\author{Chi Zhang}
\affiliation{
  \department{Computer Science}              %% \department is recommended
  \institution{University of Pittsburgh}            %% \institution is required
  \country{USA}                    %% \country is recommended
}
\email{raymond.chizhang@gmail.com}          %% \email is recommended

\author{Eddy Z. Zhang}
\affiliation{
  \department{Computer Science}              %% \department is recommended
  \institution{Rutgers University}            %% \institution is required
  \country{USA}                    %% \country is recommended
}
\email{eddy.zhengzhang@gmail.com}          %% \email is recommended

%% Abstract
%% Note: \begin{abstract}...\end{abstract} environment must come
%% before \maketitle command
\begin{abstract}
Priority queue, often implemented as a heap, is an abstract data type that has been used in many well-known applications like Dijkstra's shortest path algorithm, Prim's minimum spanning tree, Huffman encoding,and the branch-and-bound algorithm. However, it is challenging to exploit the parallelism of the heap on GPUs since the control divergence and memory irregularity must be taken into account. In this paper, we present a parallel generalized heap model that works effectively on GPUs. We also prove the \textit{linearizability} of our generalized heap model which enables us to reason about the expected results. We evaluate our concurrent heap thoroughly and show a maximum 19.49X speedup compared to the sequential CPU implementation and 2.11X speedup compared with the existing GPU implementation \cite{He+:HiPC12}. We also apply our heap to single source shortest path with up to 1.23X speedup and 0/1 knapsack problem with up to 12.19X speedup.
\end{abstract}

%% 2012 ACM Computing Classification System (CSS) concepts
%% Generate at 'http://dl.acm.org/ccs/ccs.cfm'.
\begin{CCSXML}
<ccs2012>
<concept>
<concept_id>10011007.10011006.10011008</concept_id>
<concept_desc>Software and its engineering~General programming languages</concept_desc>
<concept_significance>500</concept_significance>
</concept>
<concept>
<concept_id>10003456.10003457.10003521.10003525</concept_id>
<concept_desc>Social and professional topics~History of programming languages</concept_desc>
<concept_significance>300</concept_significance>
</concept>
</ccs2012>
\end{CCSXML}

\ccsdesc[500]{Software and its engineering~General programming languages}
\ccsdesc[300]{Social and professional topics~History of programming languages}
%% End of generated code

%% Keywords
%% comma separated list
%\keywords{}  %% \keywords are mandatory in final camera-ready submission

%% \maketitle
%% Note: \maketitle command must come after title commands, author
%% commands, abstract environment, Computing Classification System
%% environment and commands, and keywords command.
\maketitle

\section{Introduction}

A priority queue is an abstract data type which assigns each data element a priority and an element of high priority is always served before an element of low priority. A priority queue is dynamically maintained, allowing a mix of insertion and deletion updates. Well-known applications of priority queue include Dijkstra's shortest path algorithm, Prim's minimum spanning tree, Huffman encoding, and the branch-and-bound algorithm that solves many combinatorial optimization problems. 

Understanding how to accelerate priority queue on many-core architecture has profound impacts. A comprehensive study will not only shed light on the performance benefits/limitation of running the priority queue itself on accelerator architecture but also pave the road for future work that parallelizes a large class of applications which build on  priority queues.  

Priority queue is often implemented as \emph{heap}. In this paper, we focus on \emph{heap}.  Heap is a fundamental abstract data type, but has not been extensively studied for its acceleration on GPUs. A heap is a tree data structure. Using \emph{min-heap} as an example, every node in the binary tree has a key that is smaller than or equal to that of its parent. There are two basic operations for heap -- insertion updates and deletion updates. The deletion update always returns the minimal key. The insertion update inserts a key to the right location in the tree. Both operations allow logarithmic complexity and leave the binary tree in a balanced state. An example of min-heap is shown in Fig. \ref{fig:minheapexample}.

\begin{figure}[htb]
  \centering
  \includegraphics[width=0.6\linewidth]{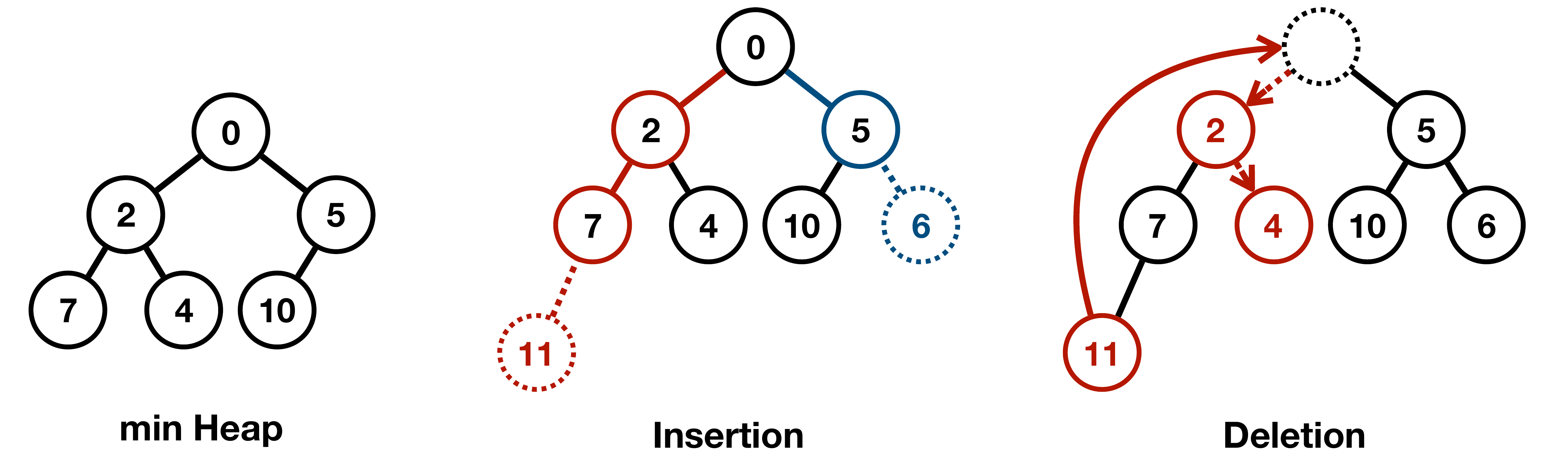}
  \caption{An example of min-heap and irregularity in the heap operations }
  \label{fig:minheapexample}
\end{figure} 

However, it is non-trivial to exploit parallelism of the binary heap, and reason about the correctness given a parallel implementation. There are two key challenges that prevent us from fully taking advantage of the massive parallelism in many-core processors. Each update operation of the binary heap involves a tree walk. Different tree walks exhibit different control flow paths and the memory locality might be low. The control divergence and memory irregularity are two main performance hazards for GPU computing \cite{Zhang+:ASPLOS11}. The two performance hazards need to be tackled before we can efficiently accelerate binary heap on GPUs. Moreover, as the parallel design gets complicated, it is not easy to reason about the correctness properties of the concrete implementation. 

%As far as we can tell, there are very few studies that both deliver an efficient implementation and in the meantime provide correctness guarantee. 

Existing work for parallelizing heap neither provide correctness guarantee nor take the GPU performance hazards into consideration. Among the research for parallelizing heap on CPUs, the work by Rao and Kumar \cite{Nageshwara+:TC88} avoids locking the heap in its entirety, associates each node with a lock, and makes insertion updates top-down, such that the locking order of nodes prevents deadlock. Hunt \etal ~\cite{Hunt+:IPL96} adopts the same fine grained locking mechanism, but makes insertion updates bottom up while maintaining a top-down lock order. The implementation by Hunt \etal ~alleviates the contention at the root node. However, neither of them formally reasons about the correctness of their implementation. Neither tackles the control divergence problem caused by random tree walks as it was not a problem for CPUs at the time when both works were published.

The closest related work to ours is by He and others \cite{He+:HiPC12} which is a GPU implementation of the binary heap. It is based on the idea presented by Deo and Prasad \cite{Deo+:JS92} in 1992, which exploits the parallelism by increasing the
node capacity in the heap. One node may contain $k$ keys ($k \ge $ 1). However, while it exploits intra-node parallelism, inter-node parallelism is not well exploited. It divides the heap into even and odd levels and uses barrier synchronization to make sure operations on two types of levels are never processed at the same time. It assumes all insertion/deletion updates progress at the same rate. Moreover, between every two consecutive barrier synchronization points, only one insertion or deletion request can be accepted, which severely limits the efficiency of its implementation on GPUs. Our implementation is shown to be much faster than the work by He \etal \cite{He+:HiPC12} (in Section 5).

In this work, we present a design of concurrent heap that is well-suited for many-core accelerators. Although our idea is implemented and evaluated on GPUs, it applies to other general purpose accelerator architecture with vector processing units. Further, we not only show that our design outperforms sequential CPU implementation and existing GPU implementation, but also prove that our concurrent heap is linearizable. Specifically, our contribution is summarized as follows: 
 
\noindent \emph{\bf \textit{1. We develop a generalized heap model}.} In our model, each node of the heap may contain multiple keys. This similar to the work by Deo and Prasad  \cite{Deo+:JS92}. However, there are two key differences. First, assuming $k$ is the node capacity, Deo and Prasad \cite{Deo+:JS92} only allow inserting/deleting exactly $k$ keys, while it is not uncommon that an application inserts/deletes less than or more than $k$ keys. We added support for partial insertion and deletion in our generalized model. Second, we exploit both intra-node parallelism and inter-node parallelism, the latter of which is not fully explored by Deo and Prasad \cite{Deo+:JS92} or He \etal ~\cite{He+:HiPC12}. Note that the benefit of having multiple keys in one node is multi-fold. It allows for intra-node parallelism, memory parallelism, local caching, and can alleviate control divergence since in a tree walk $k$ keys in the same node move along the same path. 

\noindent {\bf \textbf{2. We prove the linearizability of our implementation.} }  We propose two types of heap implementations and prove both are linearizable. Linearizability is a strong correctness condition. A history of concurrent invocation and response events is linearizable if and only if some (valid) reordering of events yield a legal sequential history. We provide a model for describing the concurrent and sequential histories and for inserting linearization points. 
%The model is used for proving the correctness of the two heap implementations, but it can also be extended to other future heap designs which use fine-grained lock. 
As far as we know, existing heap implementations on CPU \cite{Hunt+:IPL96, Nageshwara+:TC88, Deo+:JS92} or GPU \cite{He+:HiPC12} do not have a formal reasoning about their correctness or linearizability condition.

\noindent {\bf \textbf{3. We perform a comprehensive evaluation of the concurrent heap.}} We thoroughly evaluate our heap implementation and provide an enhanced understanding of the interplay between heap parameters and execution efficiency. We perform sensitivity analysis for heap node capacity, partial operation percentage, concurrent thread number, and initial heap utilization. We explore the difference between insertion and deletion performance. We also evaluate our implementation on real workloads, while most previous work use synthetic traces \cite{Hunt+:IPL96, Nageshwara+:TC88, Deo+:JS92, He+:HiPC12}, as far as we know. We show that performance improvement could be up to 19.49 times compared with sequential CPU implementation, 2.11 times compared with the existing GPU implementation. We improve the single source shortest path algorithm by up to 123\% and improve the performance of 0/1 knapsack by up to 1219\%, which demonstrates the great potential of applying priority queue on GPU accelerators.

%[TODO minheap example and batched heap put together]

\section{Background}

\subsection{Heap Data Structure}

A heap data structure can be viewed as a binary tree and each node of the binary tree stores a \textit{key}. Without loss of generation, we use the \emph{min-heap} to describe our idea throughout the paper. The minimal key is stored at the root node. A node's key is smaller than or equal
to parent's key. A heap is maintained using two basic operations: \emph{insert} and \emph{delete-min}. 

During the \emph{insertion} process, a key is inserted to an appropriate location such that the heap property is maintained. In a bottom-up insertion process, it places the key in the first empty leaf node, then repeats the following steps: compare a node's key with its parent node's key, if smaller, then swap itself with the parent node, otherwise, terminate. The \emph{bottom-up insertion} algorithm is shown in Fig. \ref{fig:serialheap} (a).

\begin{figure*}[htp]
\centering
\includegraphics[width=0.7\textwidth]{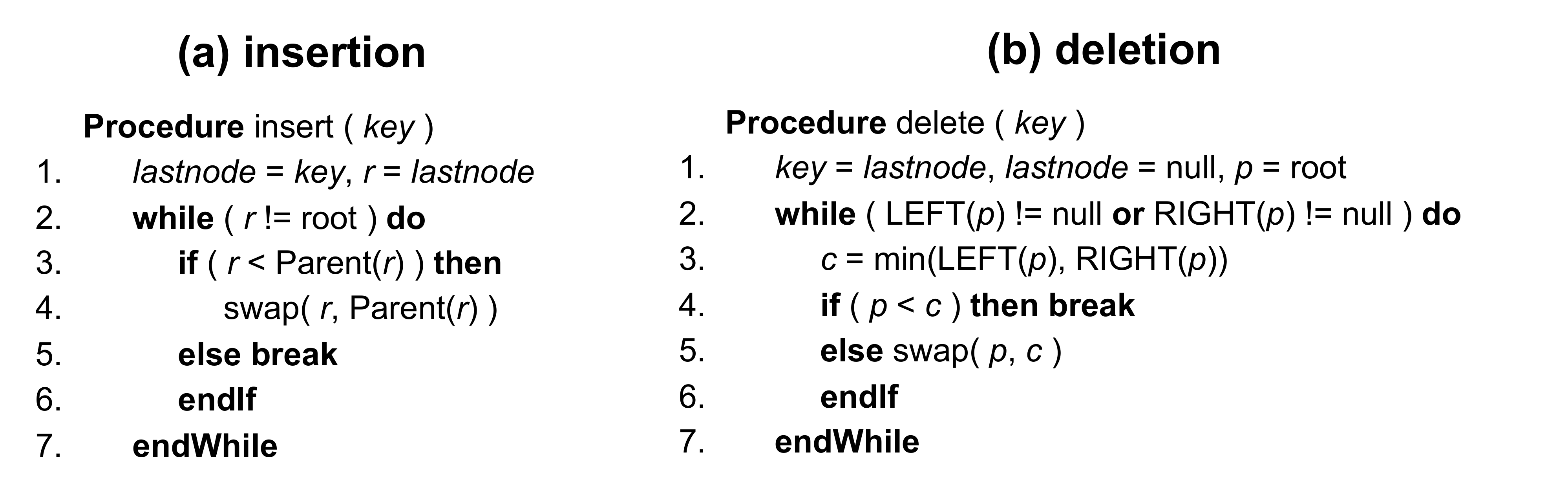}
\caption{Insert and Delete in a sequential Heap}
\label{fig:serialheap}
\end{figure*}

%[TODO variable name in the figure may be changed to the one in the text]

A \emph{delete-min} procedure returns the minimal key in the heap. It removes the key at the root node and starts a ``heapify" process to restore the \emph{min-heap} property. To heapify, it moves the last leaf node to the root, and
repeat the following steps: (1) compare the node p's left child and right child (if there is any), (2) return the smaller of the two as $c$, and (3) if p's key is larger than c's key, swap the node c with the node p, otherwise, terminate. The algorithm is shown in Fig. \ref{fig:serialheap} (b).

\subsection{GPU Architecture}

GPU is a type of many-core accelerator. It employs the single instruction multiple thread (SIMT) execution model. In order to take advantage of GPU, two fundamental factors need to be taken into consideration \cite{Kumar+:ISCA08, Moscovici+:PACT17}:
\emph{control divergence} and \emph{memory locality}.

In SIMT model, threads are organized into groups, each of which executes in lock-step manner. Each group is called a \emph{warp}. The threads in the same warp can only be issued one instruction at one time. If threads in the same warp need to execute different instructions, the execution will be serialized. This is called \emph{{control divergence}}. 

During execution, the data must be fetched for all threads at each instruction. The warp cannot execute until the data operands of all its threads are ready. Memory parallelism needs to be exploited since data in physical memory is organized into large contiguous blocks. Data is fetched in the unit of memory blocks. If one memory references involves non-contiguous data access in multiple blocks, the warp needs to fetch multiple blocks. If one memory reference involves contiguous data access in memory, it will reduce the number of memory blocks that need to be fetched.

The SIMT model provides a limited set of synchronization primitives. Barrier synchronization is allowed for threads within CTA. A GPU \emph{kernel} does not complete until all its threads
have completed, which can be used as an implicit barrier among all threads. Although there is no provided lock intrinsics on GPUs, the atomic compare and swap (CAS) function is provided ,and can be used to implement synchronization primitives.

 \subsection{Linearizability}
 
In concurrent programming, \textit{linearizability} \cite{herlihy1990linearizability} is a strong condition which constrains the possible output of a set of interleaved operations. It is also a safety property that enables us to reason about the expected results from a concurrent system \cite{shavit+:DC16}. The execution of these operations results in a \textit{history}, an ordered sequence of invocation and response events. The \emph{invocation} refers to the event when an operation starts. The \emph{response} refers to the event when operation completes. 

A \textit{sequential history} is the one that an invocation is always followed by a matching response, and a response event followed by another invocation. Alternatively, since in a \emph{sequential} history operations do not overlap, we can consider as if an invocation and its matching response happen at the same time, and an operation takes immediate effect. There is no real concurrency in a \emph{sequential} history. However, a sequential history is easy to reason about. 

We say that the \textit{history} $H$ as an ordered list of invocation and response events \{ ${e_{1}, e_{2}, ..., e_{k}}$\} is \textit{linearizable} if there exists a re-ordering of the events such that (1) a correct sequential history can be generated, and (2) if the response of an operation $e_i$ precedes the invocation of another operation $e_j$, $res(e_i)$  $<$ $inv(e_j)$, then the operation $e_i$ precedes the operation $e_j$ in the reordered events. Typically \emph{linearization point} is used to denote the time when an operation takes immediate effect between the invocation and response of one operation. Finding the right linearization points to construct a correct sequential history naturally meets the condition of (2).

\section{Concurrent Heap Design}

We exploit the parallelism of heap operations by allowing concurrent insert operations, \textit{INS}, and delete operations, \textit{DEL}, on different tree nodes. In the meantime, each node in the binary tree is extended to contain a batch of keys instead of only one. We refer to this proposed heap as the \textit{generalized heap} throughout this paper. Our heap is well-suited for acceleration on GPUs. Parallelism exists within a batch of keys and the control divergence is reduce because all keys in the same batch move along the same path in the tree during tree traversals. 

\subsection{Generalized Heap} \label{sec:design:batchedheap}

Each node in the generalized heap contains $k$ keys. \footnote{we use $k$ to represent the node capacity throughout this paper.} Since the number of keys in the heap is not always a multiple of $k$ and an insertion or deletion may not be exactly $k$ keys, we use a \textit{partial buffer} implementation. The \textit{partial buffer} contains no more than $k-1$ keys 
%which means all the keys in it are sorted and all the entries that have no keys are filled up with a MAX value
. All the keys in the partial buffer should be larger than or equal to the keys in the root node so as to make sure the smallest $k$ keys are in the root node. We denote the $k$ keys in one heap node as a \emph{batch}. 

Like conventional heap, after each \textit{INS} and \textit{DEL} update on the generalized heap, the {heap property} needs to be preserved. Here, we formally define the \textbf{generalized heap property}:

\begin{enumerate}[label=\textbf{Property \arabic*}.,itemindent=*]
  \item Given any node $c$ in the generalized heap and its parent node $p$, the smallest key in $c$ is always larger than or equal to the largest key in $p$:
  \[ \min\limits_{i=1..k}{node[c][i]} \geq \max\limits_{j=1..k}{node[p][j]}\]

  \item Given any node $c$ in the generalized heap, the keys in $c$ are sorted in ascending order:
  \[ \forall i \in [1, k):node[c][i] \leq node[c][i+1] \]

  \item Given the partial buffer $b$ with size $s$, all the keys are sorted in ascending order and are larger than or equal to those in the root node $r$:
  \[ \forall i \in [1, s):node[b][i] \leq node[b][i+1] \]
  \[ {node[b][1]} \geq {node[r][k-1]} \]
\end{enumerate}

Note that the heap property for the conventional heap is a special case of the \textbf{generalized heap property} with $k$ = 1. When the batch of each node contains only one key, the generalized heap property 1 and 2 are still satisfied. The generalized heap property 3 does not apply since the partial buffer contains at most $k - 1$ keys so there is no partial buffer in the conventional heap.

The most space efficient way to represent a generalized heap is using the \emph{array}. Each entry of the array represents a single key and consecutive $k$ entries represent a node. Thus, the generalized heap can be represented as a linear array while the first $k$ entries are from the root node and the next $k$ entries are from the second node and so on. Therefore, array entries in the range of $[i*k, (i+1)*k-1]$ are from the $i$-th node in the generalized heap. An array representation allows an implicit binary tree representation of the generalized heap. Fig. \ref{fig:batchedheap} shows an example of the generalized heap in both array representation and binary tree representation. The partial buffer is stored separately. 

\begin{figure}[htb]
  \centering
  \includegraphics[width=0.4\linewidth]{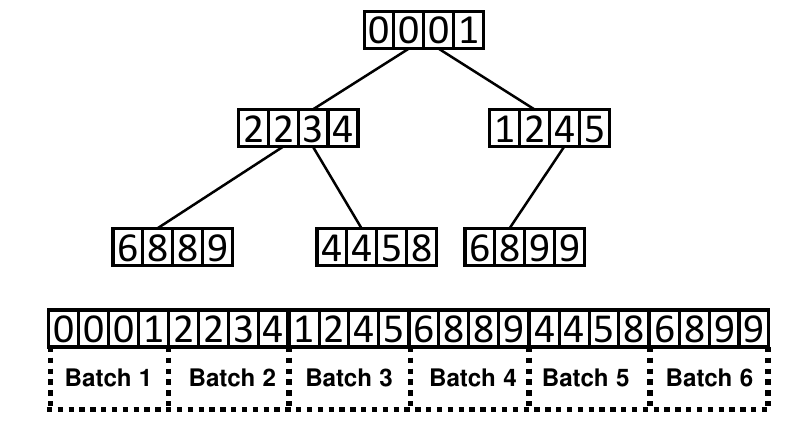}
  \caption{An example of the generalized heap}
  \label{fig:batchedheap}
\end{figure}

\subsection{\textit{INS} and \textit{DEL} operations on the Generalized Heap}

There are two basic operations for heap: \textit{DEL} operation deletes the root node which contains the smallest $k$ keys from the heap and \textit{INS} operation inserts new keys into the heap. We describe these two basic operations on the generalized heap.

\begin{figure}[htb]
  \centering
  \includegraphics[width=1.0\linewidth]{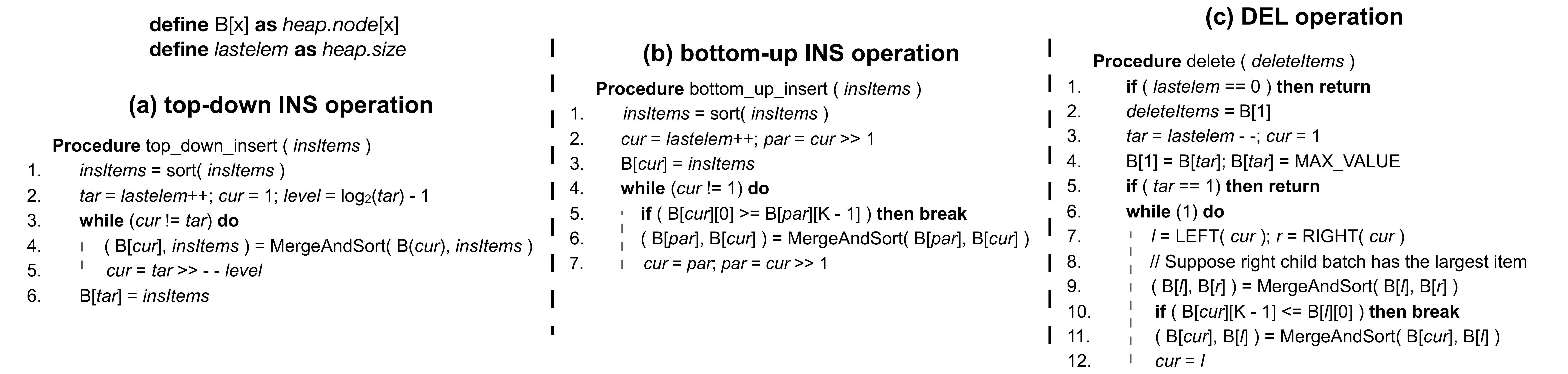}
  \caption{Pseudo Codes for \textit{INS} and \textit{DEL} operations on Generalized Heap}
  \label{fig:seqopr}
\end{figure}

\subsubsection{\textbf{\textit{DEL} Operation}} 

The \textit{DEL} operation on the generalized heap retrieves the $k$ keys from the root node. Since the root node is left empty, the generalized heap needs to be heapified to satisfy the generalized heap property. The pseudo code is shown in Fig. \ref{fig:seqopr}(c).

The \textit{DEL} operation refills the root node with the keys from the last leaf node of the heap (line 1 - 5). Note that we will fill the last node with a MAX value to make sure the old keys in that node are covered. Then, we propagate the new values in root node down. During the propagation, \textit{DEL} operation will perform the MergeAndSort operation on two child nodes $l$ and $r$ (line 9). Here we formally define the \textbf{MergeAndSort} operation between two batch of keys $a$ and $b$:

\[ (batch[hi][1:k], batch[lo][1:k]) = \textbf{MergeAndSort}(batch[a][1:k], batch[b][1:k]) \textbf{ such that}\]
\[ \forall i \in [1,k): batch[hi][i] \leq batch[hi][i+1] \]
\[ \forall i \in [1,k): batch[lo][i] \leq batch[lo][i+1] \]
\[ \max\limits_{i=1..k}{batch[hi][i]} \leq \min\limits_{j=1..k}{batch[lo][j]} \]

MergeAndSort operation returns two batches $hi$ and $lo$ with size $k$. $hi$ stores the $k$ smallest keys and $lo$ stores the $k$ largest keys. 

  The \textit{DEL} operation places the $lo$ part back to the child node whose maximum key was larger (compared with the other child) before (line 9). In this example, let's suppose it is the child node $r$. It can be proved that with such a placement policy, the generalized heap property on the sub-heap of $r$ will be maintained. The $hi$ part is placed into the child node $l$. Then, another MergeAndSort operation is applied to the current node and the child $l$ (line 11). The $k$ smallest of the merged result will stay in the current node, and the $k$ largest will propagate through sub-heap of $l$. The propagation ends until the generalized heap property is satisfied (line 10). An example of the \textit{DEL} operation is shown in Fig. \ref{fig:heapdelete}. 

\begin{figure*}[htb]
  \centering
  \includegraphics[width=0.8\textwidth]{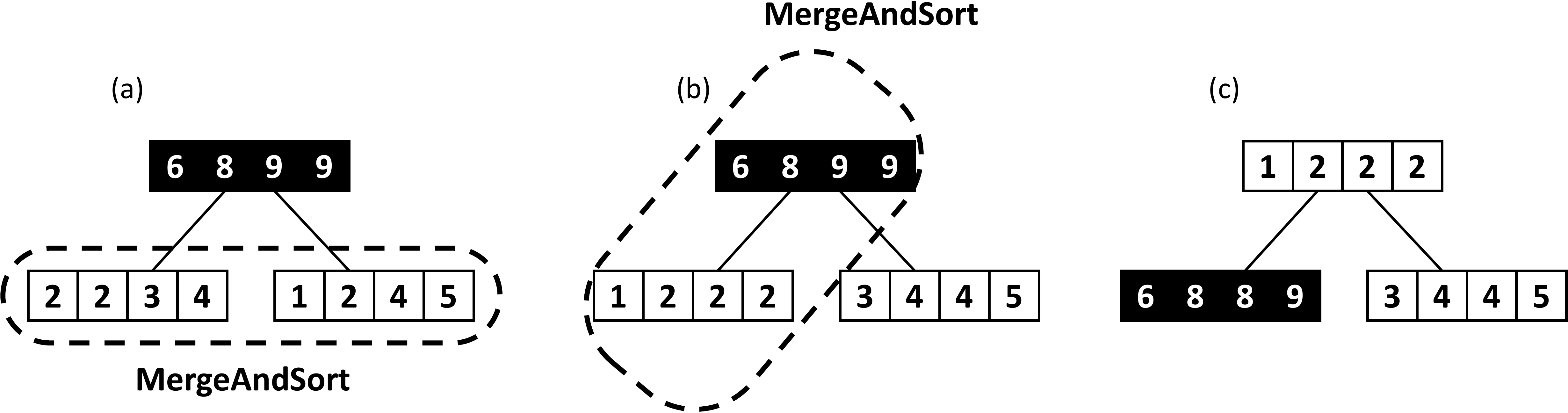}
  \caption{Example: Deletion for the generalized heap}
  \label{fig:heapdelete}
\end{figure*}

\subsubsection{\textbf{\textit{INS} Operation}}

The \textit{INS} operation inserts new keys into the generalized heap.\footnote{We suppose one \textit{INS} operation inserts at most $k$ new keys. For the case that inserting more than $k$ keys, multiple \textit{INS} operations can be invoked.} It grows the heap by adding a new node to the first empty node in the heap, we call this location the \textit{target} node. Given the \textit{target} node, a path from the root node to the target node can be found and we call it the \textit{insert path}. 

There are two possible directions for the propagation of the \textit{INS} operation, which leads to two different types of \textit{INS} operation: \textcircled{1} \textit{top-down INS}, which starts from the root node and propagates down until it reaches the target node; \textcircled{2} \textit{bottom-up INS}, which starts at the target node and propagates up until it reaches the root node or when the generalized heap property is satisfied in the middle of the heap.

\paragraph{\textbf{\textit{Top-down INS} operation}} The \textit{top-down INS} operation starts at the root node and propagates down to the bottom level of the heap. The propagation of the \textit{top-down INS} operation follows the insert path, the MergeAndSort operation is performed between the new insert items and each node on the insert path until it reaches the target node. An example of the \textit{top-down INS} operation is shown in Fig. \ref{fig:heaptopdown} and the pseudo code is provided in Fig. \ref{fig:seqopr}(a).

\begin{figure*}[htb]
  \centering
  \includegraphics[width=0.8\linewidth]{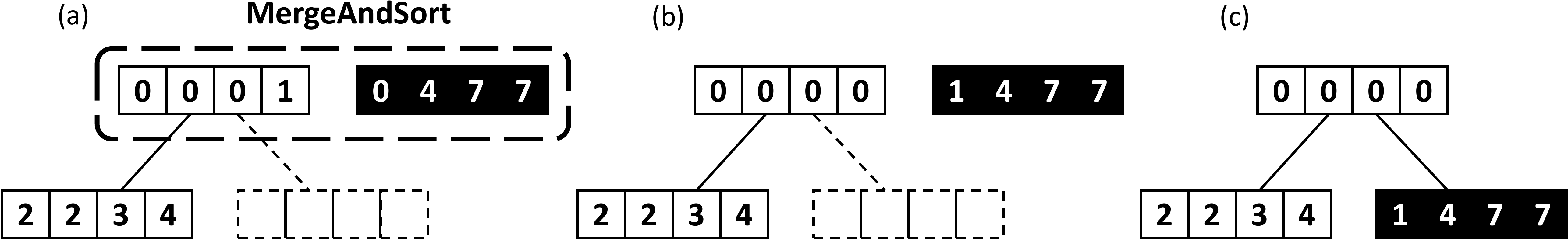}
  \caption{Example: Top-Down Insertion}
  \label{fig:heaptopdown}
\end{figure*}

\paragraph{\textbf{\textit{Bottom-up INS} operation}}
The \textit{bottom-up INS} operation inserts new keys from the bottom of the heap to the root batch of the heap and still follows the corresponding insert path. The pseudo code is shown in Fig. \ref{fig:seqopr}(b). We first move the to-insert new keys to the target node (line 3). Since the generalized heap property may be violated, the MergeAndSort operation is performed between the nodes on the insert path and their parent nodes. The propagation keeps going along the insert path until it reaches the root batch or the generalized heap property is satisfied in the middle (line 5). An example of \textit{bottom-up INS} is provided in Fig. \ref{fig:heapbottomup}.

\begin{figure*}[htb]
  \centering
  \includegraphics[width=0.8\linewidth]{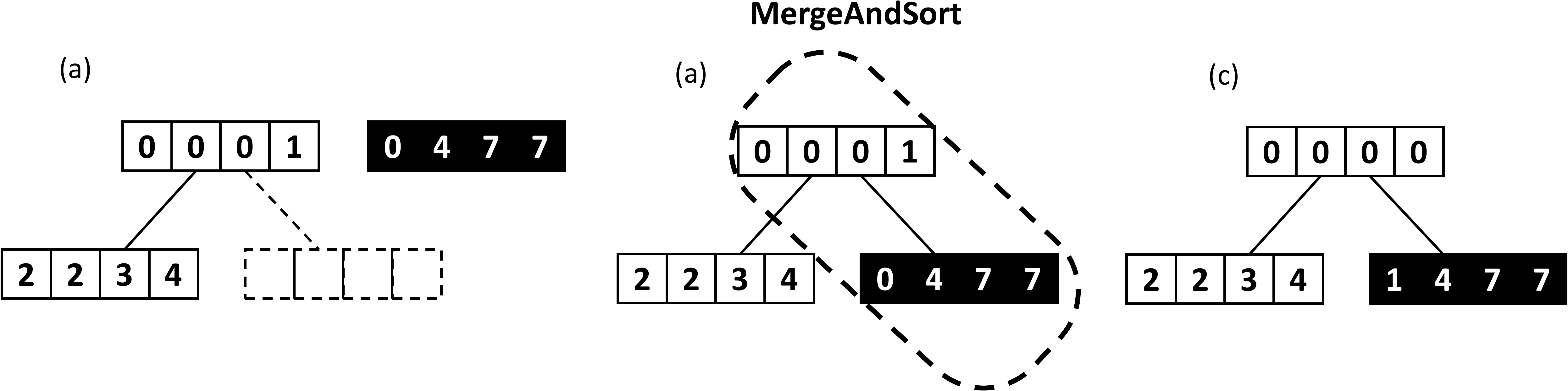}
  \caption{Example: Bottom-Up Insertion}
  \label{fig:heapbottomup}
\end{figure*}

\paragraph{\textbf{Discussion}} 

Compared to \textit{top-down INS} operation, \textit{bottom-up INS} operation may not need to traverse all the nodes on the insert path since the generalized heap property may be satisfied in the middle. Moreover, when concurrent \textit{INS} and \textit{DEL} operations are performed on the generalized heap, \textit{bottom-up INS} operation can reduce the contention on the top levels of the heap. However, the \textit{bottom-up INS} operation needs to pay attention to the potential deadlock problem caused by the opposite propagation directions of \textit{INS} and \textit{DEL} operations. We will discuss more about these concurrent \textit{INS} and \textit{DEL} operations in the following sections.

\subsection{Concurrent Heap}

In this section, we describe how \textit{DEL} and \textit{INS} operations can be performed concurrently on our parallel generalized heap. Our algorithms are inspired by the methods discussed in \cite{Nageshwara+:TC88} and \cite{Hunt+:IPL96} which introduced concurrent \textit{INS} and \textit{DEL} operations on a heap with $k=1$, with \textit{top-down INS} and \textit{bottom-up INS} operations respectively. In this paper, we call the concurrent heap with \textit{\textbf{T}op-\textbf{D}own \textbf{INS}} operation the \textit{TD-INS/TD-DEL Heap} and the one with \textit{\textbf{B}ottom-\textbf{U}p \textbf{INS}} as \textit{BU-INS/TD-DEL Heap}.

\subsubsection{Lock Order for \textit{INS} and \textit{DEL} operations} \label{sec:design:lock}

In \cite{Nageshwara+:TC88} and \cite{Hunt+:IPL96}, to support concurrent \textit{INS} and \textit{DEL} operations while ensuring correctness and avoiding deadlocks, a simple lock strategy is applied. Instead of locking the whole heap, each node of the heap is assigned a lock and only a small portion of the nodes are locked at one time. Our first implementation adopts this method. In Fig. \ref{fig:lockorder}, we show how we handle the locking order for both \textit{INS} and \textit{DEL} operations. The partial buffer is handled when the root node is locked, so that both the root node and the partial buffer is protected by the same lock. For all other nodes, each node has only one lock.

\begin{figure}[htb]
  \centering
  \includegraphics[width=1\linewidth]{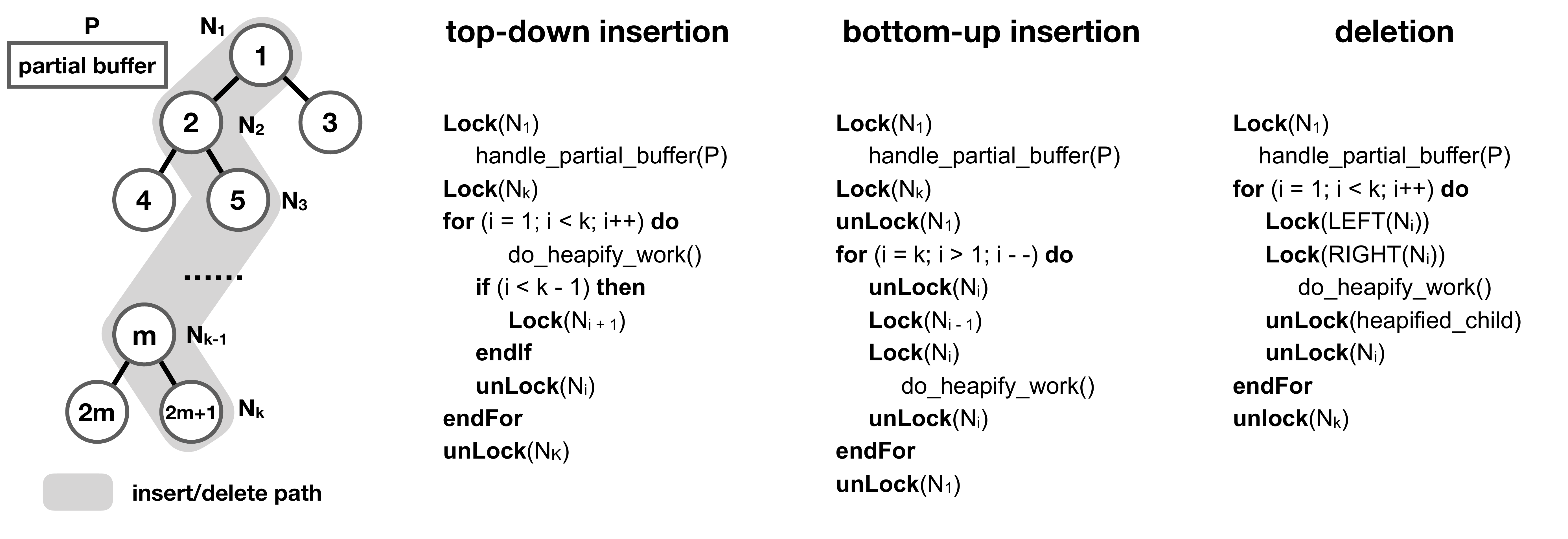}
  \caption{Lock Order for \textit{top-down INS}, \textit{bottom-up INS} and \textit{DEL} operations}
  \label{fig:lockorder}
\end{figure} 

The \textit{top-down INS} operation starts at the root node $N_1$ and propagates along the insert path $N_{1}, N_{2},...,N_{k}$. It will do the heapify work of $N_i$ when it has locked $N_i$. After it finishes its work, it will lock $N_{i+1}$ before it unlocks $N_i$ which follows a parent-child locking order. Similarly, the \textit{bottom-up INS} operation also follows the parent-child locking order. When the \textit{bottom-up INS} is at $N_i$, it will release the lock on $N_i$ first, locks $N_{i-1}$ next and then locks $N_i$. Note that, in this case, the \textit{bottom-up INS} operation does not lock any node after it releases $N_i$. The \textit{DEL} operation needs to lock more nodes during its propagation. When it is at $N_i$ and $N_i$ is locked, it then locks its two child nodes and do the heapify work. After the work is done, it unlocks the child node that is already heapified and then $N_i$. In this way, both \textit{INS} and \textit{DEL} operations follow the parent-child order so that no deadlock could happen.

 Each node of the heap is associated with a \textit{multi-state} lock. This multi-state lock has multiple states which can indicate the status of each node. The multi-state lock for \textit{top-down INS} and \textit{bottom-up INS} operations have different states. We describe the difference in the following sections.

We implement the \textit{multi-state} lock using atomicCAS. Atomic operations are well optimized on GPUs \cite{KeplerWhitePaper} which makes it a straightforward choice to implement the \textit{multi-state} lock.
 
\subsubsection{\textbf{TD-INS/TD-DEL Heap}} \label{sec:design:tditdd}

Our TD-INS/TD-DEL heap implements top-down insertions and top-down deletions, using the locking order as described in Fig. \ref{fig:lockorder}. We let the \textit{multi-state} lock have four different states: \textbf{AVAIL} indicates that the node is available; \textbf{INUSE} means that the node is acquired by another operation; To lock a node, the state of that node changes from \textbf{AVAIL} to \textbf{INUSE}. A node with state \textbf{TARGET} represents that the node is the target node of a insert operation; The state of a node is changed to \textbf{MAKRED} only when the target node is needed by a delete operation for insertion and deletion cooperation. Finite State Automata is shown in Fig. \ref{fig:fsa}.

The \textit{INS} and \textit{DEL} operations can cooperate to speedup the propagation\cite{Nageshwara+:TC88}. When the \textit{DEL} operation needs to fill the root node, if there is an \textit{INS} operation that is being in progress. The \textit{DEL} operation does not need to wait until the last leaf node to be ready (if it is not ready and if it is the target node of an in-progress insertion), it can fill the root node with the keys from the \textit{INS} operation. 

The \textit{DEL} operation changes the state of the last node from \textbf{TARGET} to \textbf{MARKED}. to let the \textit{INS} operation know that a \textit{DEL} operation is asking for the insert keys. When the \textit{INS} operation finds that the state of the target node is \textbf{MAKRED}, it moves the insert keys to the root node and terminate. The \textit{DEL} operation can then continue and use those keys in the root node for propagation.

\begin{figure}[htb]
\includegraphics[width=0.8\linewidth]{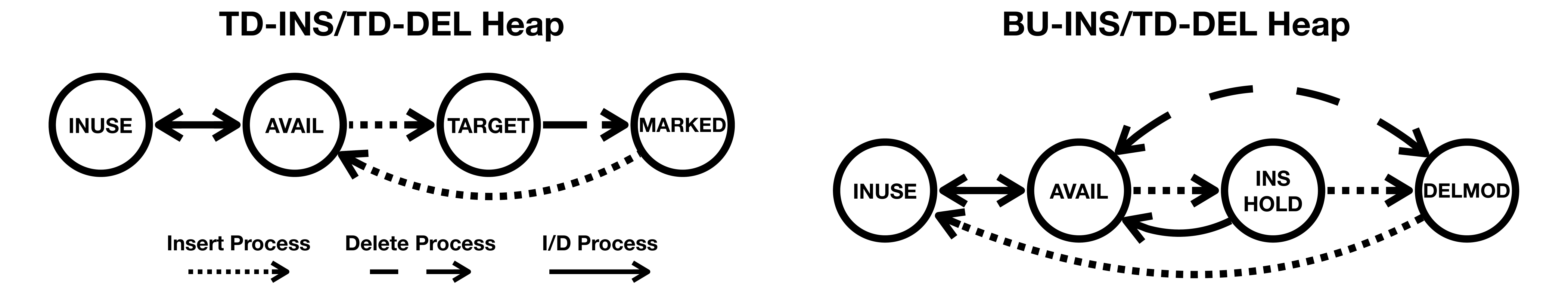}
\caption{FSA of TD-INS/TD-DEL Heap and BU-INS/TD-DEL Heap}
\label{fig:fsa}
\end{figure}

To handle partial batch insertion, we acquire the partial buffer at the time when we hold the root node. Since only one operation can work on the same node, this can make sure that no two operations can work on the partial buffer at the same time. Then we apply the MergeAndSort operation between the insert keys and the partial buffer. We check if the partial buffer have enough space to contain those new keys. If so, we perform another MergeAndSort operation between the partial buffer and the root node to satisfy the generalized heap property 3. If not, we obtain the k smallest keys from the MergeAndSort result as a full batch and propagate it down through the root node, while leaving the rest keys in the partial buffer. 

For \textit{DEL} operation, we only consider deleting the items from the partial buffer when the total number of keys is less than a full batch, in another word, all heap nodes are empty. This is because based on the generalized heap property 3, the root node always have the smallest keys in the heap. Although allowing partial batch insertion will cause extra overhead, however, the inserted keys in the partial buffer do not need to propagate into the heap immediately until the partial buffer is overflown. In this way, we still gain the benefit of memory locality and the intra-node parallelism. 

We show the pseudo code of \textit{INS} and \textit{DEL} operations on the \textit{TD-INS/TD-DEL Heap} in Fig. \ref{fig:TopDown}.

\begin{figure}[htb]
  \centering
  \includegraphics[width=1\linewidth]{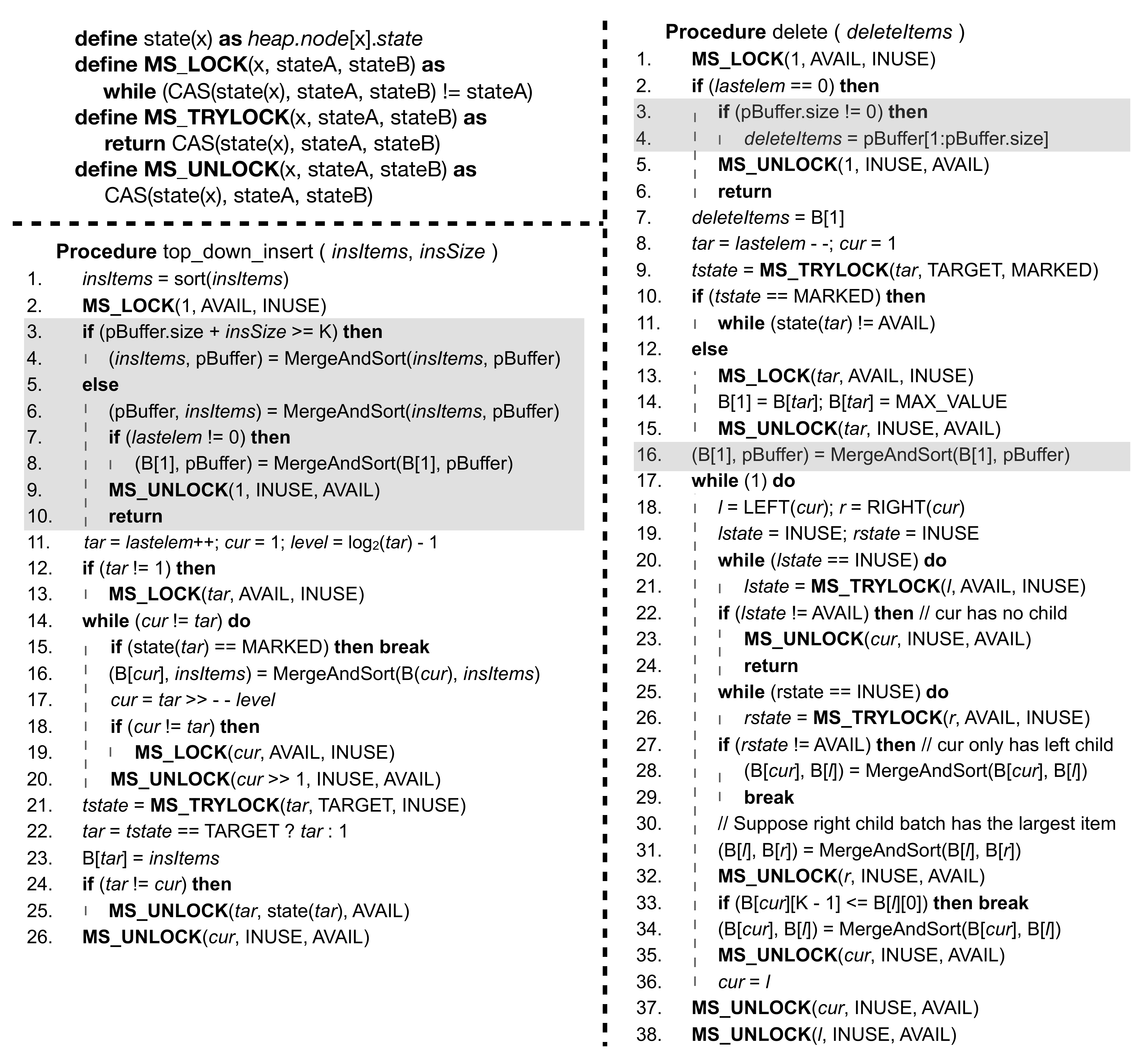}
  \caption{Pseudo Codes for \textit{TD-INS/TD-DEL Heap} Operations}
  \label{fig:TopDown}
\end{figure}

\subsubsection{\textbf{Linearizability of \textit{TD-ins/TD-del Heap}}} \label{sec:design:tdlin}
% the Top-Down concurrent insertion and deletion consist of multiple threads communicate through the same heap. It will have a set of results if some thread slightly changed order. However, sequential operations do not have any overlap and its result is fixed, also each operation's pre-condition is previous one's post-condition which makes proof of correctness really easy. we want to transfer our complex concurrency operations to sequential operations and maintain the post-condition(Heap property), this method is called linearization ({\color{red} cite the paper by Maurice Herlihy, use (shortened) formal definition from their pape}). And we call the concurrent options is \textbf{linearizable}.

We show that the heap with top-down insertion and top-down deletion (TD-Ins/TD-Del) is linearizable. 

In order to reason about the linearizability, we need to define our notations. An \emph{ins} or \emph{del} operation takes a certain amount of time to complete. We denote the time an operation is invoked as the invocation time, the time when an operation is completed as the response time. A history includes an ordered list of invocation and response events (ordered with respect to time). 

Our TD-Ins/TD-Del implementation uses fine-grained locks that each node is associated with a lock. We denote the time a thread acquires the lock of a node as \emph{acquire} time and the time a thread release the lock of a node as \emph{release} time. 

We denote an operation with a 4-tuple followed by two parameters $op[s, ac, re, t] (x) T$. The symbol $op$ is the type of the operation, \emph{ins} or \emph{del}; $s$ is the invocation time, $t$ is the response time; $ac$ refers to \emph{acquire} time of a node of interest; $re$ refers to \emph{release} of the same node; $x$ is the parameter of the operation, if the operation is insertion, it means insert $x$ into the heap, if the operation is deletion, it means $x$ is returned; $T$ refers to the thread id. Note that both $ac$ and $re$ are within the time interval $s$ and $t$, and that $ac < re$. 

To prove linearizability, we need to show that for any given history H, we can find a correct sequential history S based on a valid reordering of invocation and response events in H. Here the term ``valid reordering" refers to the case when there are two operations $e_0$ and $e_1$ in H, if the response time of $e_0$ is before the invocation time of $e_1$, $e_0$ will proceed $e_1$ in the sequential history. 

To prove such a sequential history exist, we prove the following lemma first. 
\begin{lemma}%{Fibration}
No two threads can work on the same heap node simultaneously.
\label{lem:notwothread}
\end{lemma}
\begin{proof}
According to our implementation, if a thread T has acquired the node $B$ which means the T has changed $B$'s state to INUSE, then no other thread can acquire $B$ until T releases it. 
\end{proof}

We denote a history H with n operations as $H$ = \{ $op^H_i$($s_i$, $acR_i$, $reR_i$, $t_i$) $x^H_i$ $T^H_i$ | $1 \le i \le n $ \} \footnote{This is slightly different from traditional notation of a history, but means the same.}. Here $acR$ and $reR$ respectively refer to the acquire and release of the lock for the root node in the heap. In our notation, the history $H$ is an ordered list such that its operations are ordered with respect to the time the root node is released. Since each operation in TD-INS/TD-DEL heap needs to acquire the root at its first step in fig. \ref{fig:TopDown}, and based on Lemma \ref{lem:notwothread}, only one thread can successfully acquire root node at one time. Thus for two operations $op^H_{u}$($s_{u}$, $acR_{u}$, $reR_{u}$, $t_{u}$ ) $x^H_{u}$ $T^H_{u}$, and $op^H_{v}$($s_{v}$, $acR_{v}$, $reR_{v}$, $t_{v}$) $x^H_{v}$ $T^H_{v}$, we have $ u < v $ if and only $reR_{u} < acR_{v}$.  

\begin{theorem}
The TD-INS/TD-DEL heap is linearizable. 
\label{topdown}
\end{theorem}
\begin{proof}
We show that we can construct a sequential history S given any H. To construct the sequential history, we first construct a list of linearization points \{ $N_i$ | i = 1 to n \} such that  $acR_i$ $<$ $N_i$ $<$ $reR_i$. Simply put, $N_i$ is an arbitrary time point between every pair of events that acquire and release the root node. An example of setting linearization point is shown in Fig. \ref{fig:topdownlinear}. 

An operation appear to occur instantaneously at its linearization point. A linearization point has to be between the invocation time and response time for an operation, since acquiring and releasing root is a step within every update in TD-Ins/TD-Del heap, the $N_i$ time points we set here is naturally between invocation and response time. 

Next we construct a sequential history as 

\begin{center}
$S$ = \{ $op^S_i$($N_i$) $x^S_i$ $T^S_i$ | $1 \le i \le n $ \}
\end{center}

We set a one-to-one correspondence between the i-th operation $op^H_i$($s_i$, $acR_i$, $reR_i$, $t_i$) $x^H_i$ $T^H_i$ in H and the i-th operation $op^S_i$($N_i$) $x^S_i$ $T^S_i$ in S. We set $T^S_i = T^H_i$ and $op^S_i = op^H_i$. For all insert operations, we set $x^S_i = x^H_i$, which means we insert the same key items for the corresponding operation in S. Next we prove that $x^S_i = x^H_i$ for every delete update, which means every delete operation in S returns the same value as its corresponding delete operation in H. 

Assume that when the m-th operation in H releases root node at time $reR_m$, the heap value is $Heap^H_m$, its set of keys are denoted as $set(Heap^H_m)$, and its root node is denoted as $root(Heap^H_m)$. Similarly, at the time $N_m$ of the m-th operation in S, assume the heap value is $Heap^S_m$, its set of keys are denoted as $set(Heap^S_m)$, and its root node is denoted as $root(Heap^S_m)$. We prove two properties:
\begin{enumerate}
    \item[L1:] $set(Heap^H_i) = set(Heap^S_i)$, $1 \le i \le m$
    \item[L2:] $root(Heap^H_i) = min(set(Heap^H_i))$ and $root(Heap^S_i) = min(set(Heap^S_i))$, $1 \le i \le m$
\end{enumerate}

We prove properties L1 and L2 by induction. Initially we have $H_0$ which is an empty history, and a heap value $Heap^H_0$. We set $S_0$ empty and set $Heap^S_0$ to be the same as $Heap^H_0$. Properties L1 and L2 are satisfied for the initial heap. 

Assume at the time point $reR_k$ in H and at the time point $N_k$ in S, properties L1 and L2 hold. We just need to prove for the time point $reR_{k+1}$ in H and the time point $N_{k+1}$ in S, properties L1 and L2 hold as well. There are  two cases.  

\textbf{Case I} -- the (k+1)-th operation in H is an insertion: $ins^H_{k+1}(s_{k+1}, acR_{k+1}, reR_{k+1}, t_{k+1}) x^{H}_{k+1} T^{H}_{k+1}$. At the time $reR_{k+1}$, since when a thread releases a root node in our implementation, the new item $x^{H}_{k+1}$ is already merged with the original root and the smaller item of the merged result is kept in root while the larger item may or may not propagate down the heap. Therefore the root node should contain the smallest item after taking the new item $x^{H}_{k+1}$ into consideration. Formally, $set(Heap^H_{k+1}) = set(Heap^H_{k+1}) \bigcup x^{H}_{k+1} $ and $root(Heap^H_{k+1}) = min(set(Heap^H_{k+1}))$. 

In the sequential history, since we also set (k+1)-th operation as insert update at time $N_{k+1}$, the operation is $ins^S_{k+1}(N_{k+1}) x^{S}_{k+1} T^{H}_{k+1}$, where $x^{S}_{k+1} = x^{H}_{k+1}$. In sequential history, the insertion incurs as if instantaneously, thus $set(Heap^S_{k+1}) = set(Heap^S_{k}) \bigcup x^{H}_{k+1} $ and $root(Heap^S_{k+1}) = min(set(Heap^S_{k+1}))$. Thus $set(Heap^H_{k+1}) = set(Heap^S_{k+1})$, $root(Heap^H_{k+1}) = min(set(Heap^H_{k+1}))$ and $root(Heap^S_{k+1}) = min(set(Heap^S_{k+1}))$. Properties L1 and L2 hold. 

\textbf{Case II} -- If the (k+1)-th operation in H is a deletion, $del^H_{k+1}(s_{k+1}, acR_{k+1}, reR_{k+1}, t_{k+1}) x^{H}_{k+1} T^{H}_{k+1}$. In our implementation, the deletion removes the root which is $min(set(Heap^H_k))$ and re-heapifies the heap. It does not release root until root node is updated to the smallest of the remaining items, thus at the time point $reR_{k+1}$, $root(Heap^H_{k+1}) = $ $min(set(Heap^H_{k}) - min(set(Heap^H_k)))$. 

In the sequential history, we set a delete operation at the time point $N_{k+1}$, as if the delete update happens instantaneously. Then the delete update returns $min(set(Heap^S_{k}))$ which is the same as $min(set(Heap^H_{k})))$, and in the meantime, root node will be updated to $min(set(Heap^S_{k}) - min(set(Heap^S_k)))$ which is also the same as the root of $Heap^H_{k+1}$ as described above. Thus properties L1 and L2 still hold.

 Thus we have successfully constructd a sequential history S from any given history H. Therefore, the TD-INS/TD-DEL heap is linearizable.

\end{proof}

\begin{figure*}[h]
  \centering
  \includegraphics[width=0.6\linewidth]{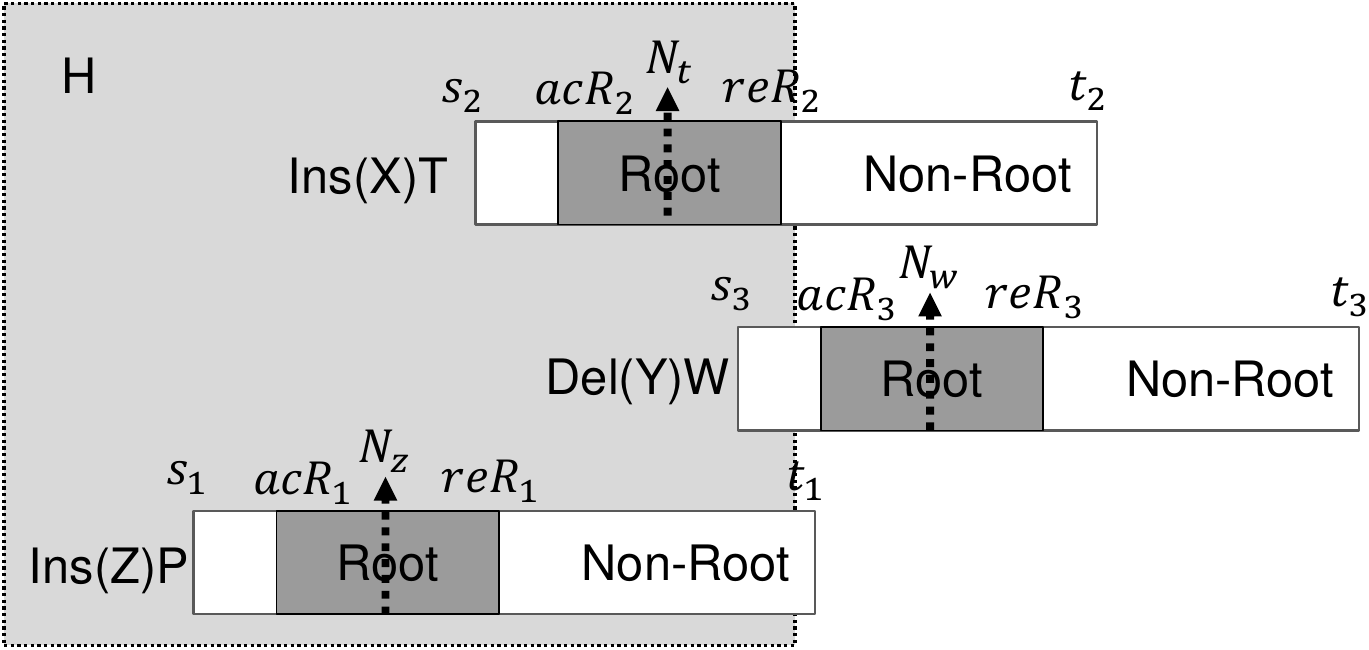}
  \caption{Example: set up of linearization points for TD-INS/TD-DEL heap operation history}
  \label{fig:topdownlinear}
\end{figure*}

\subsubsection{\textbf{BU-INS/TD-DEL Heap}}
To reduce the contention on the root node of \textit{top-down INS} operations, Hunts et.al \cite{Hunt+:IPL96} proposed a mechanism that does \textit{bottom-up INS} operations while solving the potential deadlocks from opposite propagation direction. It allows the insert thread temporarily releases the control of the insert items and a tag is used to store the $pid$ of the thread that modifies the insert items. Here we use a similar method but does not need to store the $pid$. 

% {\color{red} (Again, do not talk too much about other people's work, just abstract the lock order using another example, or some formulas, then talk about how we handle partial buffer, why the implementation fits the SIMT architecture, even when partial batches are inserted or deleted. The concurrent content is too dry. Only mention the meaning of each state at the very end after the high level intuition is described. Similarly, only mention the lengthy/tendious/complte pseudocode at the very end.)}

The \textit{multi-state} lock used in \textit{BU-INS/TD-DEL Heap} have four different states. \textbf{INUSE} and \textbf{AVAIL} are the same as the ones used in \textit{TD-INS/TD-DEL Heap}. Since the insert operation may temporarily release the control of its node, so it uses the state to tell whether the node has been modified by the time it releases the node. It changes the state of the node from \textbf{INUSE} to \textbf{INSHOLD} when it releases the node. When the insert operation acquire the node again, if it finds that the state is no longer \textbf{INSHOLD}, this means one or more delete operations have modified this batch and the insert operation can skip this batch since the delete operation makes sure the sub-heap has satisfied the \textbf{generalized heap property}. On the contrary, the delete operation which acquires the node from \textbf{INSHOLD} to \textbf{INUSE} will change the state to \textbf{DELMOD} when releasing the node.

We show the FSA of the BU-INS/TD-DEL Heap in Fig. \ref{fig:fsa}. The pseudo codes of the concurrent insert and delete operations in \textit{bottom-up} manner are shown in Fig. \ref{fig:BottomUp}. 

When an insert operation acquires the temporarily released node, there are three possible cases for the new state with that node:
\begin{enumerate}
    \item \textbf{INSHOLD}: the insert operation holds the batch successfully and the MergeAndSort operation can be performed with the parent batch.
    \item \textbf{DELMOD}: the batch has been modified by one or more delete operations, the insert operation can move to the parent batch.
    \item \textbf{INUSE}: some other operations are using this batch
\end{enumerate}

% \begin{figure}[htb]
%   \centering
%   \includegraphics[width=1\linewidth]{bottomup.pdf}
%   \caption{Concurrent Insertion and Deletion in Bottom-Up Mode}
%   \label{fig:BottomUp}
% \end{figure}

\begin{figure}[htb]
  \centering
  \includegraphics[width=1\linewidth]{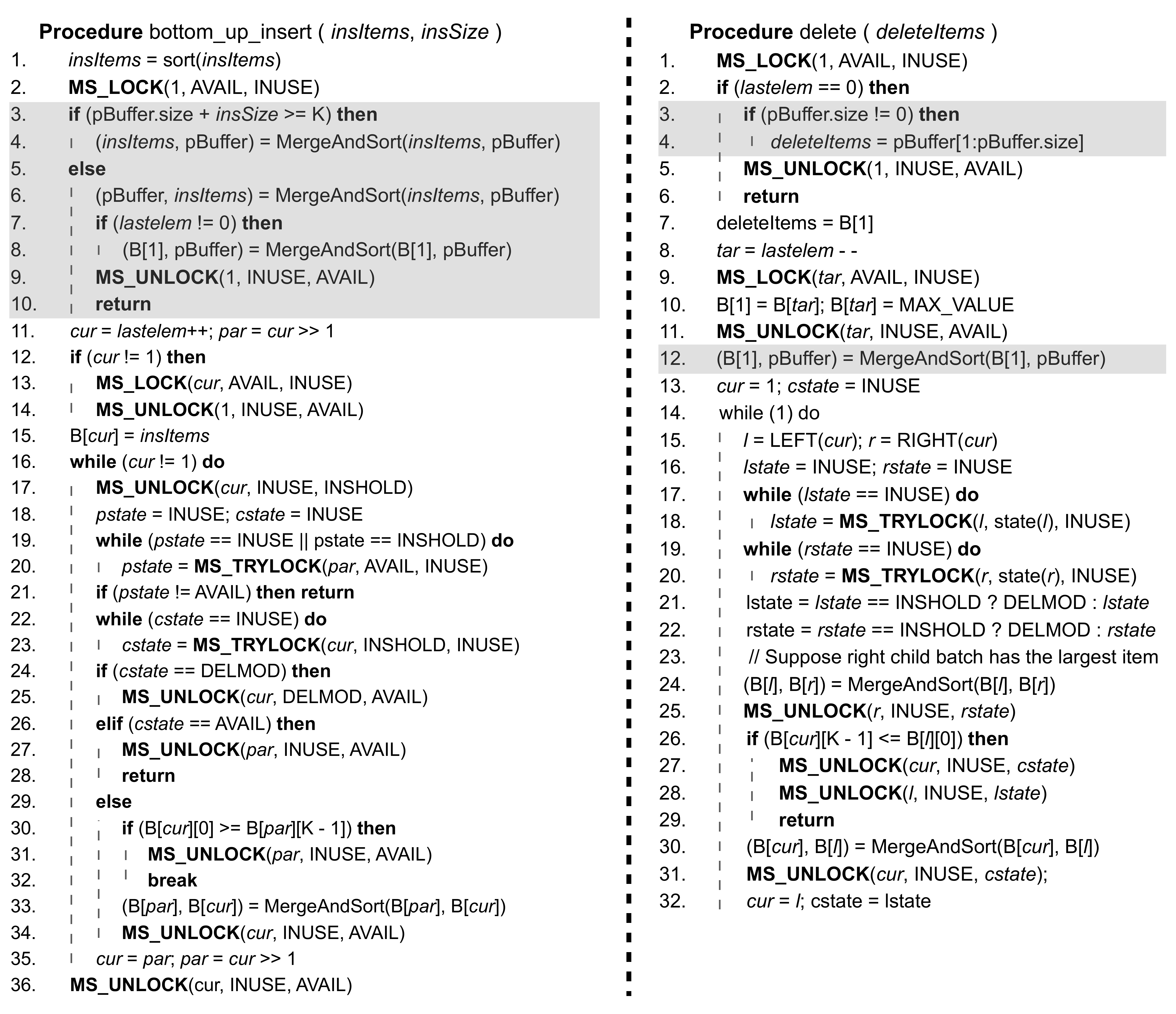}
  \caption{Concurrent Insertion and Deletion with partial batch in Bottom-Up Mode}
  \label{fig:BottomUp}
\end{figure}

The state of the parent node may also be changed. If the state is not \textbf{AVAIL} or \textbf{INUSE}, this means a delete operation has already deleted the parent node and the insert operation can terminate.

Partial buffer is handled at the beginning of each operation when it locks the root node. For both \textit{INS} and \textit{DEL} operations, the part for handling the partial buffer is exactly the same as the one we showed in Section \ref{sec:design:tditdd} so as to make sure generalized heap property 3 is satisfied.

\subsubsection{\textbf{Linearizability of BU-ins/TD-del Heap}} Now we show that the \textit{BU-INS/TD-DEL Heap} is linearizable. Note that we will use the same notations that we have defined previously in Section \ref{sec:design:tdlin}. We use the notation $ins_{i}^{H}(s_{i},acL_{i},reL_{i},t_{i})x_{i}^{H}T_{i}^{H}$ for \textit{INS} operation $i$ and $del_{j}^{H}(s_{j},acR_{j},reR_{j},t_{j})x_{j}^{H}T_{j}^{H}$ for \textit{DEL} operation $j$. Here $acL$ and $reL$ respectively refer to the acquire and release for the last locked node in a \textit{bottom-up INS} operation. This last locked node may or may not be the root node, since the generalized heap property may be satisfied in the middle of a bottom-up $INS$ update.

We denote a history H with n operations as $H = \{ op_{i}^{H}(s_{i},ac_{i},re_{i},t_{i})x_{i}^{H}T_{i}^{H}|1\leq i \leq n \}$ while $ac_i$ and $re_i$ can be either $acL_i$ and $reL_i$ for \textit{INS}, or $acR_i$ and $reR_i$ (R is for root \footnote{The same notation is already used in Section \ref{sec:design:tdlin}}) for \textit{DEL}. The operations in H are ordered with respect to the time $re_i$ ( either when \textit{INS} release the last locked node or when \textit{DEL} release the root node ). It is possible that the time $re_u$ and $re_v$ are the same for two operations $op_u$ and $op_v$. In this case, an arbitrary order can be chosen. Thus for two operations $op^H_{u}$($s_{u}$, $ac_{u}$, $re_{u}$, $t_{u}$ ) $x^H_{u}$ $T^H_{u}$, and $op^H_{v}$($s_{v}$, $ac_{v}$, $re_{v}$, $t_{v}$) $x^H_{v}$ $T^H_{v}$, we have $ {u} < {v} $ when $re_{u} < re_{v}$ or $re_{u} == re_{v}$.  

\begin{lemma}
\label{lem:btupoverlap}
Given a delete operation $del_{i}^{H}(s_{i},acR_{i},reR_{i},t_{i})x_{i}^{H}T_{i}^{H}$ and an insert operation $ins_{j}^{H}$($s_{j}$, $acL_{j}$, $reL_{j}$, $t_{j}$)$x_{j}^{H}T_{j}^{H}$ that $[acL_j,reL_j]\cap[acR_i,reR_i]$ $\neq\emptyset$, then $x_{i}^{H} \leq x_{j}^{H}$
\end{lemma}

\begin{proof}
Based on Lemma \ref{lem:notwothread} and the condition $[acL_j,reL_j]\cap[acR_i,reR_i]$ $\neq\emptyset$, we know that the last locked node of the insert operation is not the root node. Thus we can derive that $min(x_j^H) \geq max(x_i^H)$ which indicates that $x_{i}^{H} \leq x_{j}^{H}$.
\end{proof}

\begin{theorem}
The \textit{BU-INS/TD-DEL} heap is linearizable.
\label{thrm:butd}
\end{theorem}

\begin{proof}

We show that we can construct a sequential history S given any H. We construct a list of linearization points \{ $N_i$ | i = 1 to n \} such that $acR_i$ $<$ $N_i$ $<$ $reR_i$ if the i-th operation is  \textit{DEL}, or $acL_i$ $<$ $N_i$ $<$ $reL_i$ if the i-th operation is \textit{INS}. We pick $N_i$ as an arbitrary time point between the provided time range. We construct the sequential history $S$ = \{ $op^S_i$($N_i$) $x^S_i$ $T^S_i$ | $1 \le i \le n $ \} based on the linearization points. Each operation  $op^H_i$ in S corresponds to an operation $op^S_{i}$ in H. 

Like what we did in Section \ref{sec:design:tdlin}, we set $x^S_i = x^H_i$ if the i-th operation is \textit{INS}. We will prove that $x^S_j = x^H_j$ if the j-th operation is \textit{DEL}. We prove by induction. Initially we have an empty history $H_0$ and the heap value $Heap^H_0$. We set $S_0$ empty and set $Heap^S_0$ to be the same as $Heap^H_0$. At the beginning, we perform a (dummy) deletion in H and also a (dummy) deletion in S, both \textit{DEL} return the same result. The dummy deletion in H completes before any real operation starts. The heap value for S and for H after the dummy deletion will be be the same.  

Assume that $x^S_k = x^H_k$ at the time $re_k$ in H and at $N_k$ in S while the k-th operation in H is a delete operation $del^H_k$ and the k-th operation in S is also a delete operation $del^S_k$. Additionally, $set(Heap^S_k) = set(Heap^H_k)$, where $Heap^S_k$ is the heap value at the time $N_k$ in S, $Heap^H_k$ is not necessarily the heap value at the time $reR_k$ in H, rather, it is the set of keys that are contributed by all preceding insertion/deletion in H's ordered list (note that the operations in H are already ordered, see the beginning of Section 3.3.5). Let the next delete update in H be the (k+m)-th operation $op^H_{k+m}$. If we prove that at the time point $re_{k+m}$ in H and $N_{k+m}$ in S with $m\geq1$, both delete operations return the same value, that is $x^S_{k+m} = x^H_{k+m}$, and $set(Heap^S_{k+m}) = set(Heap^H_{k+m})$, then we prove that all matching delete operations in S and H return the same value by induction.

In the concurrent history $H$, between time point $reR_{k}$ and $reR_{k+m}$, there are $m-1$ concurrent operations and these operations are all insert operations. Among all these m - 1 insert operations, we let $I^{H}$ be a set of insert operations such that

$I^H$ = \{$ins_{i}^{H}(s_{i},acL_{i},reL_{i},t_{i})x_{i}^{H}T_{i}^{H}$ | $k < i < k+m $, $(acL_{i}, reL_{i})$ $ \cap (acR_{k+m}, reR_{k+m})$ $\neq \emptyset$\} 

We let the set $I'$ be the insert operations from these $m-1$ operations but not in $I$. The difference between $I$ and $I'$ is that all $I'$ operations complete before $op^H_{k+m}$, while $I$ operations might overlap with $op^H_{k+m}$. If we consider the inserted keys contributed by I' as $X^{I'} = \bigcup_{i\in I'}x^H_i $. The set of keys in the heap would be $X^{I'} \cup Heap^H_{k}$ if none of the operations in I has taken effect, the minimum would be $min(X^{I'} \cup Heap^H_{k})$. 

Let $X^{I} = \bigcup_{j \in I}x^H_j $, it is not difficult to show that $min(X^{I'} \cup Heap^H_{k}) = min(X^{I'} \cup X^{I} \cup Heap^H_{k})$.  According to Lemma \ref{lem:btupoverlap}, for any insertion $i$ in $I$, since its (acL, reL) interval overlaps with the root acquiring and releasing interval for $del^H_{k+m}$,  we know that the last node locked by operation i cannot be the root, and thus the inserted value $x^S_{i}, i \in I $ cannot be smaller than the root node. That is $x^S_{i} \ge min(X^{I'} \cup Heap^H_{k}) $ for any $i \in I$. Therefore, $min(X^{I'} \cup Heap^H_{k}) = min(X^{I'} \cup X^{I} \cup Heap^H_{k})$ is proved. The implication is that regardless if any operation or all operations in the set I complete, $op^H_{k+m}$ will always return $min(X^{I'} \cup X^{I} \cup Heap^H_{k})$ which is $min(\bigcup_{k<i<k+m} x^H_i \cup Heap^H_{k} )$.

In the sequential history $S$, there are $m-1$ insert operations between time point $N_k$ and $N_{k+m}$. The heap should include the set of keys that are in $Heap^S_{k}$ and also $x^S_i$ ($ k < i < k+m $). Thus $op^H_{k+m}$ returns $min(\bigcup_{k < i < k+m} x^S_i \cup Heap^S_{k} )$. 

Since $set(Heap^S_{k}) = set(Heap^H_{k})$ and all matching insertion operations use the same parameter value for H and S, both $op^H_{k+m}$ and $op^S_{k+m}$ return the same value. Note that it is trivial to prove that $set(Heap^H_{k+m}) = set(Heap^S_{k+m}) $

Thus we have successfully constructed a sequential history S from any given history H. Therefore, the BU-INS/TD-DEL heap is linearizable

\end{proof}

\section{Implementation}

The building blocks of the generalized heap include the sorting operation, the MergeAndSort operation, and the \textit{multi-state} lock that we introduced in Section \ref{sec:design:lock}. We use parallel bitonic sorting algorithm for local sorting operation within a thread block and merge-path\cite{Oded+:SC12} algorithm for the MergeAndSort operation. We introduce optimization to eliminate redundant MergeAndSort operations and enable an early stop mechanism to reduce the total computation load and alleviate the contention on the locks.

In our implementation, threads in one thread block work together for one \textit{INS} and \textit{DEL} operation. We choose thread-block-level operation since barrier synchronization is provided within a thread block while no built-in inter-CTA synchronization is provided and the overhead for synchronization between all thread blocks is high. 

Besides, threads within the same thread block have access to the same shared memory space, which can increase data reuse during propagation. We load frequently used items into the shared memory. Using thread-block-level \textit{INS} and \textit{DEL} operations can also benefit from memory coalescing. The items in the same node are placed continuously in the memory so that threads within the warp can achieve maximum memory coalescing. Also, the multi-state lock is a thread-block-level lock which is safer than a thread-level lock since a thread-level lock may cause deadlock due to desynchronization within a warp \cite{Wong+:ISPASS10}.

\subsection{Sorting Operation}

The \textit{INS} operations sorts the to-insert items before the propagation starts. To perform sorting, these to-insert items are loaded to the shared memory for efficient data access and movement. In our generalized heap implementation, the number of to-insert items for one insert operation is limited by the size of the shared memory per thread block (no more than 1K pairs in our case). We choose parallel bitonic sorting algorithm as it can be adopted for our thread-block-level operations well.

Bitonic sorting is a comparison-based sorting algorithm. Other efficient non-comparison based GPU sorting algorithms (e.g. parallel radix sort) require types to have the same lexicographical order as integers. This not only limits the practical use of the sorting algorithms to only numeric types like $int$ or $float$ but also the sorting complexity of which is based on the size of the key (length of the data). As we mentioned before, in our parallel generalized heap, the number of the to-insert items is usually small, which means the size of the key can dominant the sorting efficiency. Parallel bitonic sorting algorithm's complexity depends on the number of input elements which makes it more suitable for our case.

\subsection{MergeAndSort Operation} \label{sec:implement:merge}

In both \textit{INS} and \textit{DEL} operations, we perform the MergeAndSort operations frequently during the heapify process. 
%A straightforward way is to sort the items that destroys the \textbf{generalized heap property}. 
This can be optimized thanks to the \textbf{generalized heap property} 2 that the keys in a node are already sorted. Instead of directly sorting the keys that need to be merged, performing a MergeAndSort operation on those sorted nodes will be more time efficient.

In our parallel generalized heap implementation, the number of items in a node is small and we also load the data into the shared memory. Here we use the GPU \textit{merge-path} algorithm \cite{Oded+:SC12}, which merges two already sorted sequences. The main advantage of the \textit{merge-path} method is that it can assign workload evenly to threads. It has low-latency communication and high-bandwidth shared memory usage which our implementation can benefit from a lot. Detailed description and complexity analysis of the GPU \textit{merge-path} algorithm can be found in \cite{Oded+:SC12}.

\subsection{Optimizations} \label{sec:implement:optimization}

To improve the performance of concurrent \textit{INS} and \textit{DEL} operations, we apply the following optimizations.

\paragraph{Remove Redundant MergeAndSort Operations} MergeAndSort operation is the major overhead of heap operations. It is used frequently to make sure that the generalized heap satisfies \textbf{generalized heap property}. In our implementation, we compare the keys in the nodes and then decide if a MergeAndSort operation is necessary. When the largest keys in a node is smaller than the smallest key in the other, instead of performing the MergeAndSort operation, we simply swap the two nodes which is much more efficient. This optimization reduces the number of MergeAndSort operations within every insert and delete operation.

\paragraph{Early Stop} This optimization is similar to what we do in a conventional heap. The \textit{INS} and \textit{DEL} operations will terminate once the \textbf{generalized heap property} is satisfied. For our generalized heap, \textit{Early Stop} can happen for all operations except for \textit{top-down INS} which has to bring the to-insert items to the target node that locates at the bottom of the heap. Thus it has to traverse all levels of the generalized heap.

\paragraph{Bit-Reversal Permutation} The \textit{INS} operation needs to decide the target node and two consecutive \textit{INS} operations may select the two target nodes with the same parent. In this case, the insert path from the root node to the target nodes are highly overlapped which can increase the contention between the two \textit{INS} operations. We apply the bit-reversal permutation\cite{Hunt+:IPL96} that makes sure for any two consecutive \textit{INS} operations, the two insert paths have no common nodes except the root node. Consecutive \textit{DEL} operation also select the last batch in the heap following the bit-reversal permutation like \textit{INS} operation, but in the reversed order.

\section{Evaluation}

\subsection{Experiment Setup}
We perform our experiments on an NVIDIA TITAN X GPU with an Intel Xeon E5-2620 CPU with 2.1GHz working frequency. The TITAN X GPU has 28 streaming multi-processors (SMs) with 128 cores, for a total of 3584 cores. Every thread block has 48 KB of shared memory and 64K available registers. The maximal number of active threads is 1K per thread block and 2K per SM.

%[TODO explain why no delete partial]
We evaluate our parallel  heap from six different perspectives:
\begin{itemize}
    \item We compare our concurrent heap implementation with a sequential CPU heap and a previous GPU Heap \cite{He+:HiPC12}. We use input workloads with different heap access patterns.
    \item We vary the number of the number of thread blocks to evaluate the impact of contention levels and the scalabiltiy. The number of threads affects the number of active \textit{ins} or \textit{del} operations. 
    \item We perform sensitivity analysis with respect to heap node capacity K, the type of operation \textit{ins} or \textit{del}, and thread block size.
    \item We evaluate how inserting partial batches would influence the heap performance by varying the percentage of partial batch operations. 
    \item We test the concurrent \textit{ins} and \textit{del} performance under different heap utilization which means the heap is initialized with different number of pre-inserted keys.
    \item We apply our parallel heap to two real world applications which are single source shortest path (sssp) and 0/1 knapsack problem.
\end{itemize}

\subsection{Concurrent Heap v.s. Sequential Heap}

We use the GPU parallel heap implementation by He, Deo, and Prasad \cite{He+:HiPC12} as our GPU baseline. We refer to this implementation as \textit{parallel synchronous heap} or in short, \textit{P-Sync Heap}. We use the C++ STL priority queue library as the sequential CPU heap which we refer to as the \textit{STL Heap}. Note that \textit{INS} operation of the \textit{P-Sync Heap} is \textit{top-down}, while it is \textit{bottom-up} for \textit{STL Heap}.

We evaluate the performance of inserting 512M keys into an empty heap and then deleting all these 512M keys from the heap. We use different types of input keys which are \textcircled{1} randomized 32-bit int keys \textcircled{2} 32-bit int keys sorted in ascending order \textcircled{3} 32-bit int keys sorted in descending order. The results are shown in Table \ref{tab:vscpu}.

Our concurrent heap has an average 16.59X speedup compared to the \textit{STL Heap} and 2.03X speedup compared to \textit{P-Sync Heap}. We observe the best performance when the input keys are sorted in ascending order in all cases. For \textit{STL Heap} and \textit{BU-INS/TD-DEL Heap}, it is because the \textit{INS} operations only need to place the insert keys at the target node without traversing the entire heap. For \textit{P-Sync Heap} and \textit{TD-INS/TD-DEL Heap}, although \textit{INS} operations start from the root node, we still gain the benefit of the keys sorted in ascending order as it avoids the overhead of unnecessary merging operations along the insert path. 

Both \textit{TD-INS/TD-DEL Heap} and \textit{BU-INS/TD-DEL Heap} are faster than \textit{P-Sync Heap}. This is because we can support concurrent \textit{INS} or \textit{DEL} at the same level of the heap while for \textit{P-Sync Heap}, only one \textit{INS} or \textit{DEL} can work on the same level which exhibits with a lower inter-node parallelism. In later experiments, we will use randomized 32-bit keys for all \textit{INS} and \textit{DEL} operations performance evaluation except for the real world applications.

\vspace{.5cm}

%[TODO use pure insert / delete any difference?]

    \begin{minipage}{.54\columnwidth}
        
        {\scriptsize
        \centering
        \captionof{table}{Heap performance on CPU v.s. GPU }
        \label{tab:vscpu}
        \centering
        \renewcommand\arraystretch{1.5}
        \begin{tabular}{|c||c|c|c|}
            \hline
            Method & \textbf{random} & \textbf{descend} & \textbf{ascend} \\
            \hline
            \textbf{STL Heap} & 1,959,550 & 1,898,015 & 1,214,906 \\
            \textbf{P-sync Heap} & 209,648 & 205,761 & 201,66 \\
            \textbf{TD-INS/TD-DEL Heap} & 112,090 & 100,163 & 99,082 \\
            \textbf{BU-INS/TD-DEL Heap} & 104,417 & 97,593 & 96,247 \\
            \hline
        \end{tabular}
        \caption*{Thread block number: 128, thread block size: 512, \\ K: 1024, time unit: milli-second (ms), keys: 512M}
        }
    \end{minipage}
    \begin{minipage}{.4\columnwidth}
       \centering
       \captionsetup{justification=centering}
       \includegraphics[width=.9\columnwidth]{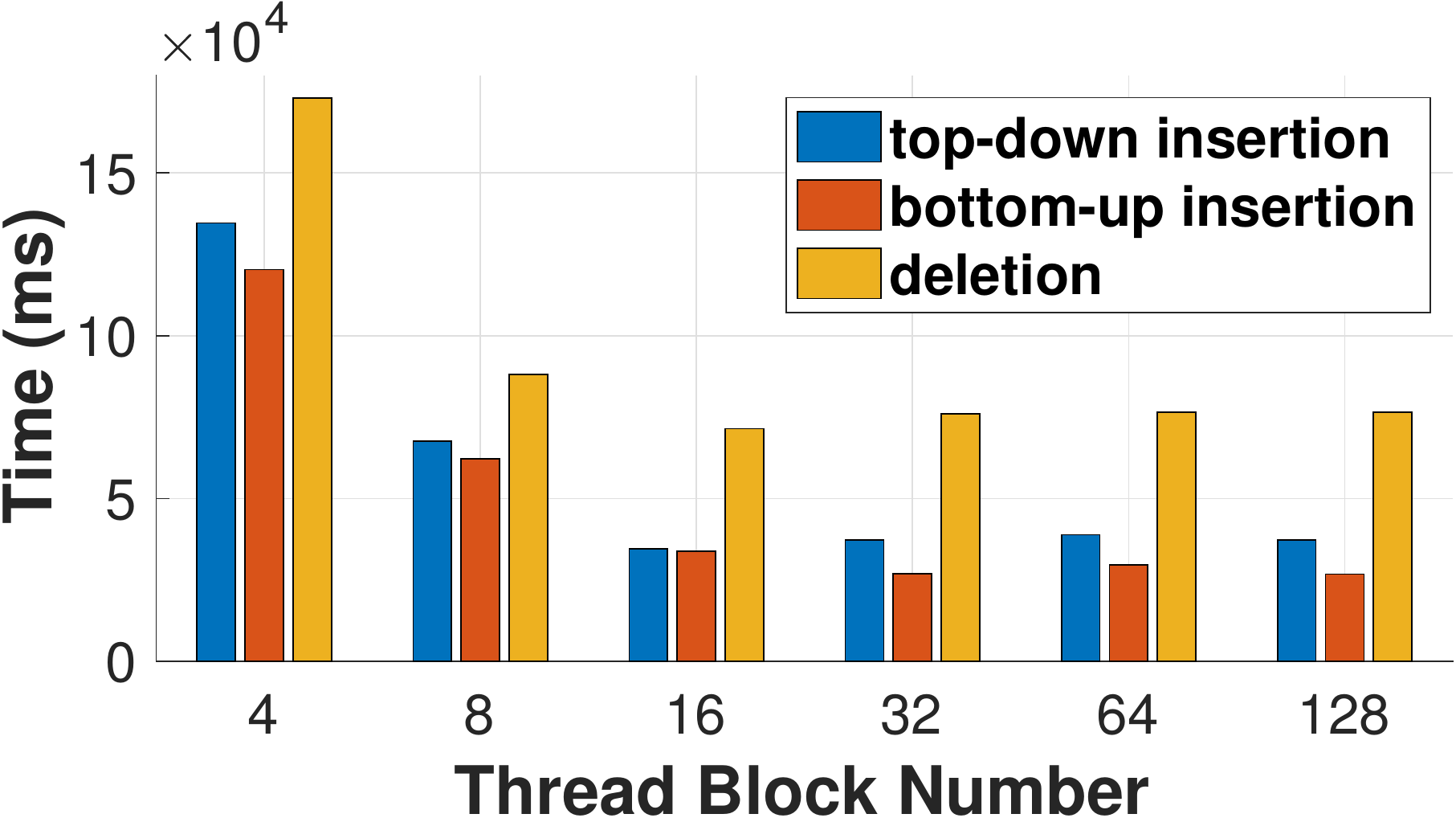}
       \captionof{figure}{Heap performance w.r.t thread block numbers}
       \label{fig:tbnum}
    \end{minipage}

\subsection{Impact of Thread Number}

We evaluate the performance of \textit{top-down}  insertion update, \textit{bottom-up} insertion update, and \textit{deletion} update respectively by varying the number of thread blocks. The more thread blocks, the more concurrency we can gain, and also the more contention on the heap. We fix all other parameters, with thread block size = 512 and batch size = 1024. We test the performance of inserting 512M random 32-bit keys into an empty heap for insertion-only experiments, and deleting 512M keys from a full heap for deletion-only experiemnt.\footnote{\noindent In this case, a fully heap is defined as a heap that has 512M keys, regardless of the batch size.} We show the results in Fig. \ref{fig:tbnum}.

The performance of both \textit{ins} and \textit{del} operations become better when the number of thread blocks is increased since more concurrency can be obtained. However, the benefit from concurrency is restricted when the thread block number keeps increasing since more thread blocks also means more contention on the heap nodes. 

The \textit{del} operation always needs much more time than the \textit{ins} operations especially when the thread block number is large which with an average 2.6X slow down. This is because \textit{del} operation needs to hold both parent node and its two child nodes and performing at most two MergeAndSort operations when updating keys on each level of the heap while \textit{ins} operation needs only one.

When comparing \textit{top-down ins} with \textit{bottom-up ins} operations. We see that \textit{bottom-up ins} always has a better performance since it causes less contention on the root node of the heap and it may not need to traverse all the nodes on the insert path (the heap property may be satisfied in the middle). 

\subsection{Impact of Heap Node Capacity}

Fig. \ref{fig:nodeS} shows how \textit{ins} and \textit{del} performance is influenced by heap node capacity. Due to the limits in shared memory size per thread block, the maximum batch size we used is 1K. Also, the maximum number of thread block size depends on the batch size, since it does not make sense to have more than one thread handling one key in MergeAndSort operations. We test the performance by inserting 512M keys to an empty heap and deleting 512M keys from a full heap.

\begin{figure}[htb]
\centering
\begin{subfigure}{.32\columnwidth}
   \centering
   \includegraphics[width=.9\columnwidth]{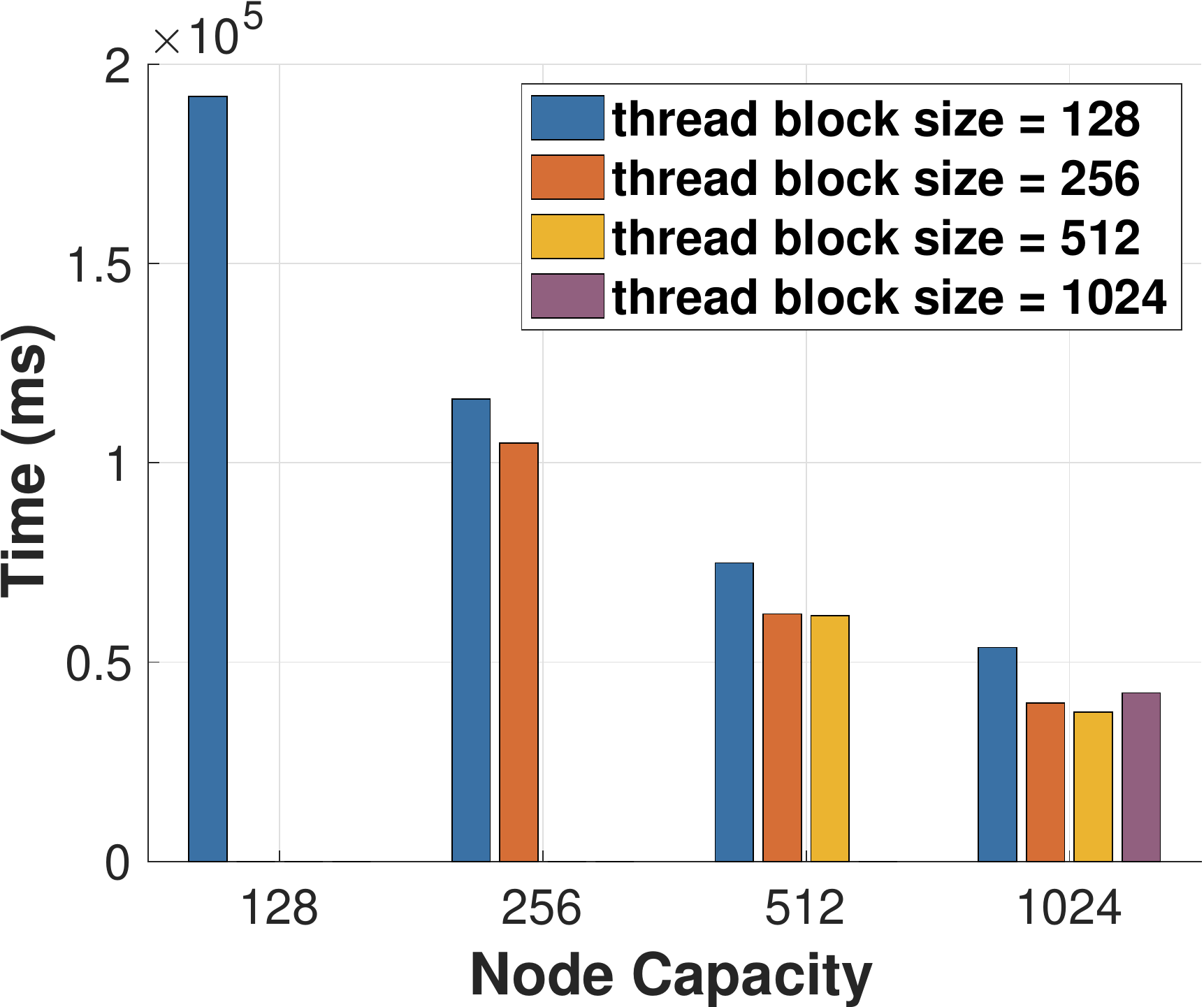}
   \captionof{figure}{top-down ins operations}
   \label{fig:nodeStd}
\end{subfigure}
\begin{subfigure}{.32\columnwidth}
   \centering
   \includegraphics[width=.9\columnwidth]{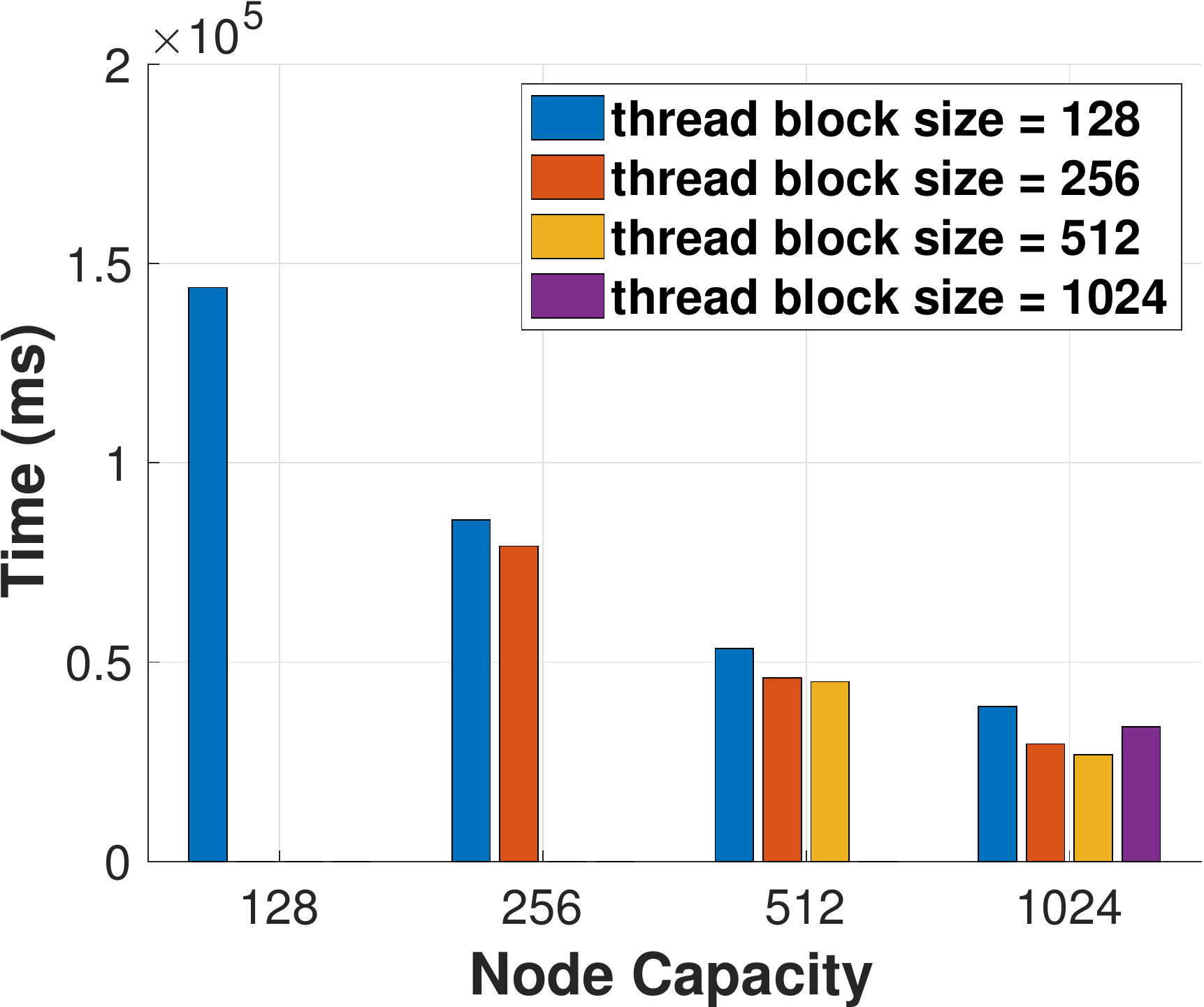}
   \captionof{figure}{bottom-up ins operations}
   \label{fig:nodeSbu}
\end{subfigure}
\begin{subfigure}{.32\columnwidth}
   \centering
   \includegraphics[width=.9\columnwidth]{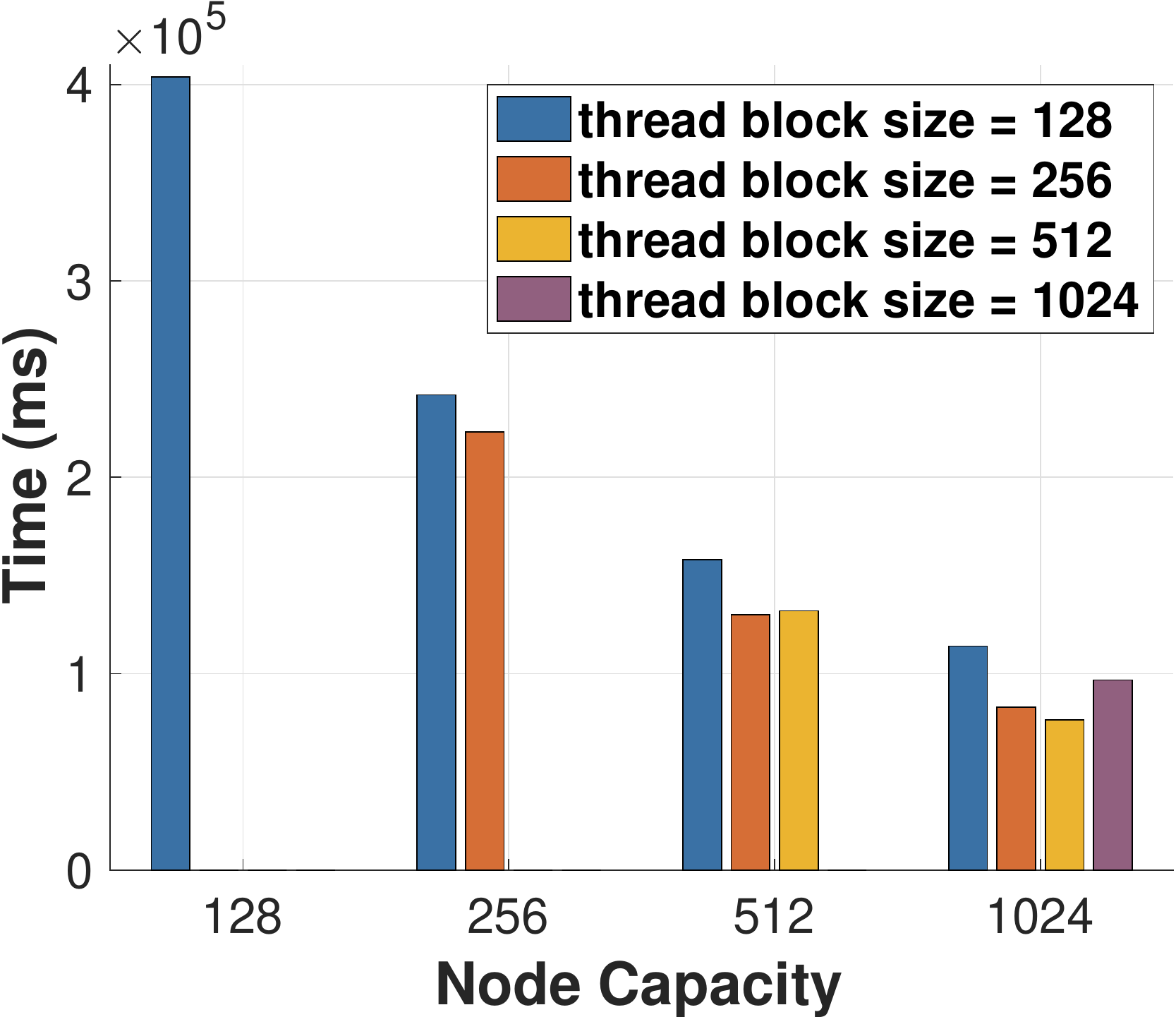}
   \captionof{figure}{del operations}
   \label{fig:nodeSdel}
\end{subfigure}
\caption{Performance of Inserting and Deleting 512M Keys w.r.t Node Capacity and Thread Block Size}
\label{fig:nodeS}
\end{figure}

When thread block size is the same, for both \textit{ins} and \textit{del}, we can observe that the performance becomes better when we use a larger node capacity. Using a larger node capacity means that with the same number of keys, the depth of the heap is reduced. If the node capacity is doubled then the level of the heap is reduced by one which leads to fewer MergeAndSort operations and tree walks. Also a larger node capacity can provide more intra-node parallelism.

In Fig. \ref{fig:nodeS}, it also shows that it is not always good to increase the thread block size because large thread block size can increase the overhead of synchronization within a thread block. Among all these configurations, we choose one with thread block size 512 and batch size 1024 for later experiments since it has the best performance for both \textit{ins} and \textit{del} operations.

\subsection{Impact of Initial Heap Utilization}
We control the initial heap size by pre-inserting a certain number of keys, for instance, to achieve a initial 10-level heap, we need to insert $batchSize * 2^{10}$ keys. In this experiment, every thread performs one \textit{ins} and one \textit{del}, which we call an \textit{ins-del} pair. Since the number of thread blocks is fixed and at most such number of \textit{ins} could happens at the same time, so the heap level is also fixed only if the initial heap level is higher than a certain number. In our experiment, we use 128 thread blocks and each thread block will perform 2K \textit{ins-del} pairs with a total 256K pairs.

In Fig. \ref{fig:util}, we show the heap performance with respect to different initial heap size from a 6-level heap with 64K items  to a 18-level heap with 256M items. We can observe that when the initial heap utilization is increased, these \textit{ins-del} pairs need more time to finish. Both \textit{ins} and \textit{del} may traverse more levels of the heap and perform more MergeAndSort operations. Operations on \textit{BU-Heap} have a better performance since \textit{bottom-up ins} still has the benefit of stopping tree traversals earlier.

\begin{figure}[htb]
\centering
\begin{minipage}{.45\columnwidth}
   \centering
   \includegraphics[width=.7\columnwidth]{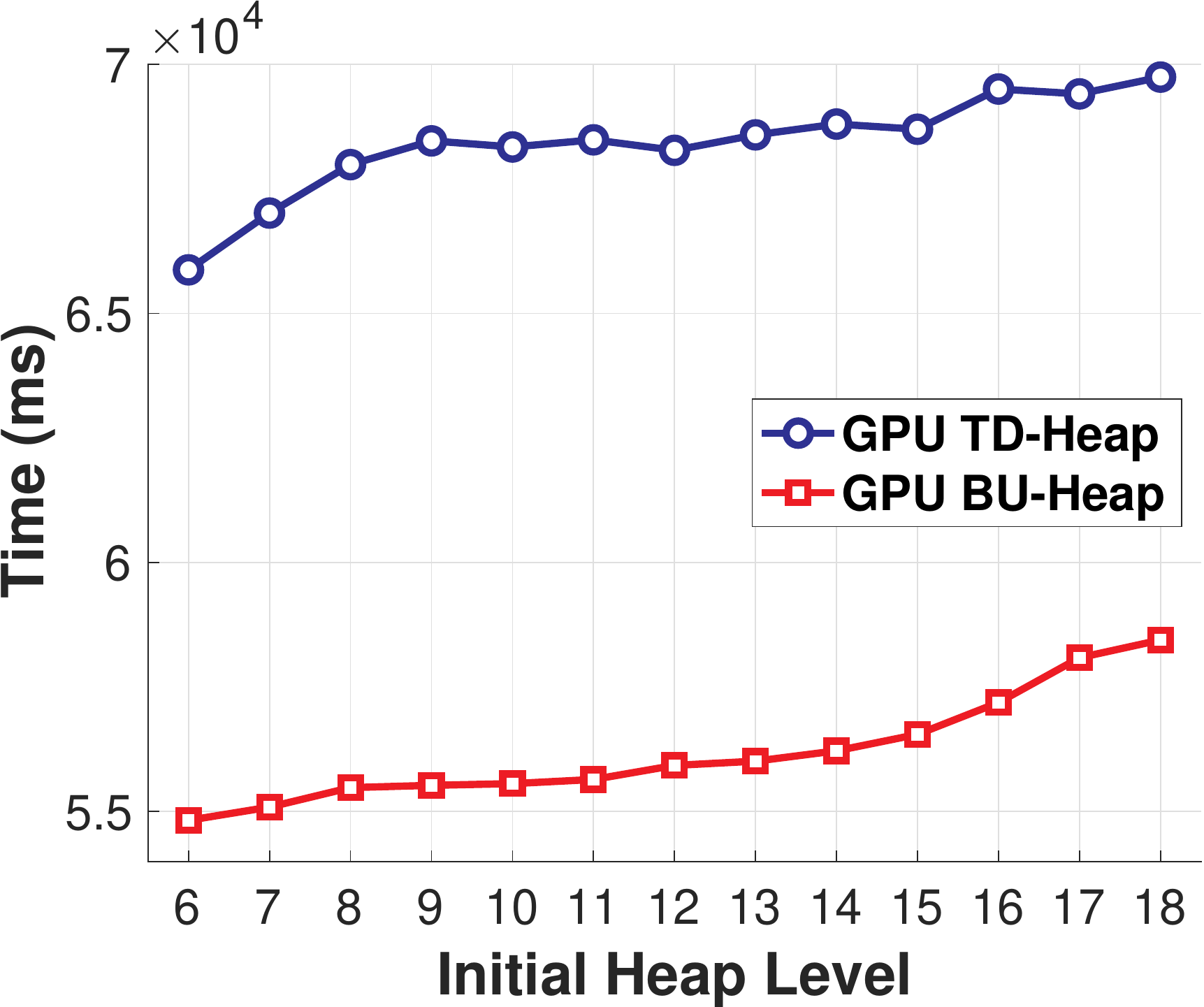}
   \caption{Performance w.r.t Initial Heap Size}
   \label{fig:util}
\end{minipage}
\begin{minipage}{.45\columnwidth}
   \centering
   \includegraphics[width=.7\columnwidth]{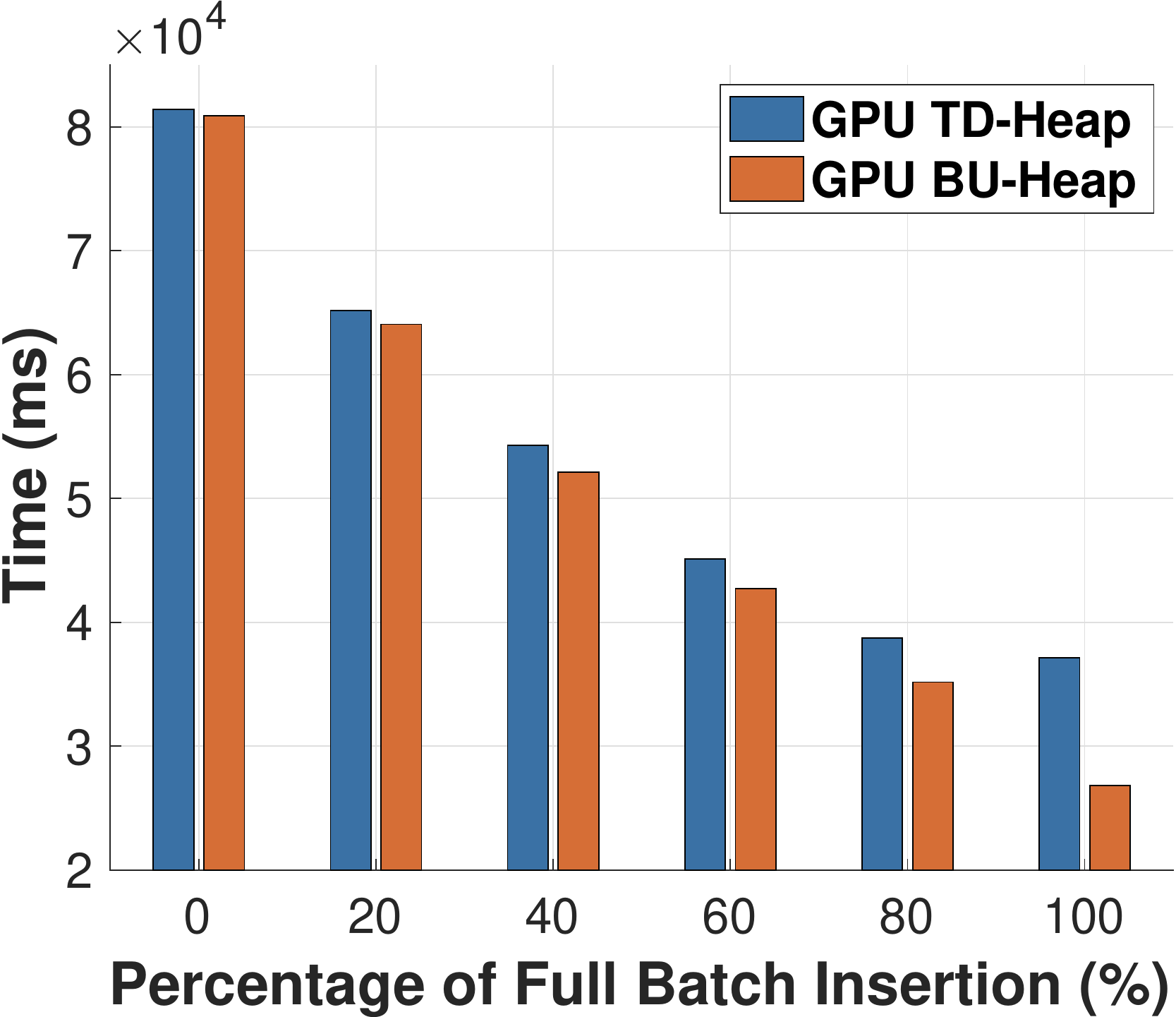}
   \caption{Performance w.r.t Partial Buffer}
   \label{fig:partial}
\end{minipage}
\end{figure}

\subsection{Impact of Partial Buffer and Partial Batch Insertion}

We evaluate how partial batch updates influence the heap performance. We test the performance by inserting 512M items into an empty heap. We control the percentage of full batch insertions and let rest insertions be randomly generated partial batches. The results are shown in Fig. \ref{fig:partial}.
%[TODO may add an experiment that inserting the same amount of batches]
As we can see, with the increase in the percentage of full batch insertions, the performance of becomes better. This is potentially because more threads are needed to insert the same number of keys, and it also cause more contention on the lock that protects the root node since inserting a partial batch will always require to lock the root node in both BU-INS/TD-DEL heap and TD-INS/TD-DEL heap. The implication is that, although partial batch is supported in the heap implementation, it would be good to avoid using partial batch insertions in real workloads if the total number of inserted keys is the same, since the overall performance difference could be up to 4X.

\subsection{Concurrent Heap with Real World Applications}
We apply our concurrent generalized heap to two real world applications: the single source shortest path (SSSP) algorithm and the 0/1 knapsack problem. Both applications can take the advantage of our concurrent heap by processing items with higher priority first. The purpose of this section is to shed light on the potential of incorporating our concurrent heap with many-core accelerators to solve real world applications. Further optimization to our concurrent heap with application-based asynchronous updates for insertion and deletion is possible, but we will leave it as our future work.

\subsubsection{SSSP with Concurrent Heap}
Gunrock\cite{Wang+:gunrock} is a well known parallel iterative graph processing library on the GPU. It applies a compute-advance model to solve for applications like the SSSP such that at each iteration, nodes are classified as active nodes and inactive nodes by checking if their new distance bring an update to existing distance, after which only active nodes would be explored in the next iteration since inactive nodes will not bring updates to the final result.

Our implementation of the parallel SSSP algorithm is similar to Gunrock's. At each iteration, we use our heap to store those active nodes with their current distance as the key. In this way, those nodes with the shortest distance would be explored first in the next iteration. As a result, our implementation tends to reduce the overhead of unnecessary updates and the number of active nodes being explored.

We use gunrock\cite{Wang+:gunrock} as the baseline for comparison and we set a threshold $N$ such that only when the number of active nodes is larger than $N$, will we incorporate the algorithm with our concurrent heap. We use N=10K in our experiments. We choose 14 different real world graphs and describe the properties of these graphs in Table \ref{tab:ssspgraphinfo}.

\begin{table}[htb]
\caption{Graph Information}
\label{tab:ssspgraphinfo}
{\footnotesize
\begin{tabular}{|c|c|c|l|}
\hline
\textbf{Graph Name} & \textbf{\# Nodes} & \textbf{\# Edges} & \textbf{Type of Graph} \\ \hline
\textbf{AS365} & 3,799,275 & 22,736,152 & 2D FE triangular meshes \\
\textbf{bundle\_adj} & 513,351 & 20,721,402 & Bundle adjustment problem \\
\textbf{coPapersDBLP} & 540,486 & 30,491,458 & DIMACS10 set \\
\textbf{delaunay\_n22} & 4,194,304 & 25,165,738 & DIMACS10 set \\
\textbf{hollywood-2009} & 1,139,905 & 115,031,232 & Graph of movie actors \\
\textbf{Hook\_1498} & 1,498,023 & 62,415,468 & 3D mechanical problem  \\
\textbf{kron\_g500\_logn20} & 1,048,576 & 89,240,544 & DIMACS10 set \\
\textbf{Stanford\_Berkeley} & 685,230 & 7,600,595 & Berkeley-Stanford web graph \\
\textbf{Long\_Coup\_dt0} & 1,470,152 & 88,559,144 & Coupled consolidation problem \\
\textbf{M6} & 3,501,776 & 21,003,872 & 2D FE triangular meshes \\
\textbf{NLR} & 4,163,763 & 24,975,952 & 2D FE triangular meshes \\
\textbf{rgg\_n\_2\_20\_s0} & 1,048,576 & 13,783,240 & Undirected Random Graph \\
\textbf{Serena} & 1,391,349 & 65,923,050 & Structural Problem \\ \hline
\end{tabular}
}
\end{table}

\begin{table}[htb]
\caption{Parallel Single Source Shortest Path Performance}
\label{tab:sssp}
{\footnotesize
\begin{tabular}{|c|c|c||c|c|c|c||c|} 
\hline
{} & \multicolumn{2}{c||}{\textbf{Baseline}} & \multicolumn{4}{c||}{\textbf{Heap Based SSSP w/ N=10K}} & {} \\
\hline
\multirow{2}{*}{\textbf{Graphs}} & \textbf{Computate}  & \textbf{\# Nodes} &  \textbf{Heap}  & \textbf{Compute}     & \textbf{Total}    & \textbf{\# Nodes} & \multirow{2}{*}{\textbf{Speedup}} \\
& \textbf{Time(ms)} & \textbf{Visited} &  \textbf{Time(ms)} & \textbf{Time(ms)} & \textbf{Time(ms)} & \textbf{Visited} &  \\
\hline
\textbf{AS365} & 654.44 & 19,664,769 & 193.43 & 422.12 & 615.55 & 11,843,368 & 1.06 \\
\textbf{bundle\_adj} & 144.54 & 903,097 & 11 & 126.48 & 137.48 & 877,675 & 1.05 \\
\textbf{coPapersDBLP} & 46.13 & 981,876 & 12.52 & 25.46 & 37.98 & 710,794 & 1.21 \\
\textbf{delaunay\_n22} & 1125.93 & 29,832,633 & 283.04 & 647.61 & 930.65 & 18,607,590 & 1.21 \\
\textbf{hollywood-2009} & 100.17 & 2,007,447 & 14.22 & 74.35 & 88.58 & 1,370,459 & 1.13 \\
\textbf{Hook\_1498} & 233.76 & 2,786,271 & 31.39 & 182.52 & 213.91 & 1,756,776 & 1.09 \\
\textbf{kron\_g500\_logn20} & 117.79 & 2,590,570 & 28.72 & 73.1 & 101.82 & 860,552 & 1.16 \\
\textbf{Long\_Coup\_dt0} & 190.06 & 2,699,927 & 43.46 & 116.77 & 160.23 & 1,571,565 & 1.19 \\
\textbf{Stanford\_Berkeley} & 55.3 & 530,294 & 5.17 & 52.54 & 57.71 & 462,860 & 0.96 \\
\textbf{M6} & 677.95 & 20,972,903 & 161.88 & 472.67 & 634.56 & 16,126,697 & 1.07 \\
\textbf{NLR} & 894.71 & 29,583,224 & 318.35 & 439.61 & 757.96 & 16,123,803 & 1.18 \\
\textbf{rgg\_n\_2\_20\_s0} & 920.27 & 7,112,685 & 46.44 & 701.78 & 748.22 & 5,871,411 & 1.23 \\
\textbf{Serena} & 124.16 & 2,594,858 & 27.98 & 84.94 & 112.93 & 1,498,836 & 1.1 \\ \hline
\end{tabular}
}
\end{table}

% Despite the efficient design of our concurrent heap, cases where the overhead of heap operation takes over its overall improvement still exists. 
% For example, at the first few iterations of the parallel SSSP algorithms, the number of active nodes might be relatively small such that processing them directly in parallel would be a better choice than reordering them through the heap as doing so may introduce extra overhead.

We show the result of parallel SSSP in Table \ref{tab:sssp}. For all the graphs we tested, we have an average of 1.13X overall speedup with the threshold $N$ = 10K compared to the baseline. The heap based sssp does not perform well on graph \textit{Stanford\_Berkeley} since it is a small graph, which means that the number of active nodes at each level is not large enough for the improvement brought by the heap to cover the overhead of it's own operations.

In Table \ref{tab:sssp}, we also list the time in milliseconds for sssp computation and heap operations separately. The computation time is the SSSP computation time, which includes processing times for node expanding, edge filtering and distance updating. The heap time is the time spent on the heap operations. The number of nodes visited represents the total number of times that nodes being explored during SSSP computation. With incorporation of our heap, the number of node visits is reduced remarkably compared to the baseline, which directly leads to the reduced sssp computation time.

\subsubsection{0/1 Knapsack with Concurrent Heap}
The knapsack problem appears in real-world decision-making processes in a wide variety of fields. It defines as follows: given weights and benefits for some items and a knapsack with a limited capacity W, determine the maximum total benefit can be obtained in the knapsack. The 0/1 knapsack problem is a branch of the knapsack problem where one must either put the complete item in the knapsack or don't pick it at all.

Branch and bound is an algorithm design paradigm, which is usually used for solving combinatorial optimization problems such as the 0/1 knapsack problem. The solution to the 0/1 knapsack problem can be expressed as a path in a binary decision tree where each level in the tree represents we either pick or do not pick an item. Thus, with n items, there are $2^{n}$ possible solutions. Instead of blindly checking for every possible solution for the maximum benefit under certain capacity, we can prune the search space by comparing the bound (the best possible benefit we could gain if we choose this node) of a node with the current maximum benefit to see if it is worth continue exploring.

A simple sequential implementation of such algorithm is to enqueue to-explore nodes into a priority queue with their current benefits as the key values and always choose to explore nodes with largest benefit first. On one hand, it is a greedy approach that would give optimal solution despite that it might encounter several local optimal before get to the global optima. On the other hand, if nodes with larger benefit are explored first, we can skip exploring certain nodes with a bound that is smaller than the max current benefit. We implement a parallel GPU knapsack algorithm based on the sequential version with our concurrent heap to show that the incorporation of our concurrent heap accelerates knapsack computation.

Since parallelizing node exploration might result in unnecessary growth in heap size, we also apply a technique which filters invalid nodes when heap size is larger than a certain threshold before the algorithm continues to explore more nodes. We name this version as knapsack with garbage collection (GC). 

\begin{table}[htb]
\caption{Datasets for 0/1 Knapsack Problem}
\label{tab:knapsackData}
{\footnotesize
\begin{tabular}{|c|c|c|c|c|}
\hline
\textbf{Dataset} & \textbf{Type} & \textbf{Size} & \textbf{Range} \\
\hline
\textbf{ks\_sc\_700\_18k} & Strongly Correlated & 700 & 18000 \\
\textbf{ks\_sc\_800\_18k} & Strongly Correlated & 800 & 18000 \\
\textbf{ks\_sc\_200\_7k} & Strongly Correlated & 200 & 7000 \\
\textbf{ks\_asc\_750\_16k} & Almost Strongly Correlated & 750 & 16000 \\
\textbf{ks\_asc\_1300\_6k} & Almost Strongly Correlated & 1300 & 6000 \\
\textbf{ks\_asc\_500\_7k} & Almost Strongly Correlated & 500 & 7000 \\
\textbf{ks\_esc\_900\_18k} & Even-odd Strongly Correlated & 900 & 18000\\
\textbf{ks\_esc\_1200\_13k} & Even-odd Strongly Correlated & 1200 & 13000 \\
\textbf{ks\_esc\_400\_8k} & Even-odd Strongly Correlated & 400 & 8000 \\
\textbf{ks\_ss\_100\_18k} & Subset Sum & 100 & 18000 \\
\textbf{ks\_ss\_1250\_12k} & Subset Sum & 1250 & 12000 \\
\textbf{ks\_ss\_1300\_14k} & Subset Sum & 1300 & 14000 \\
\bottomrule
\end{tabular}
}
\end{table}

\begin{table}[htb]
\caption{0/1 Knapsack Problem with  heap}
\label{tab:knapsackResult}
{\footnotesize
\begin{tabular}{|c|c|c||c|c|c||c|c|c|}
\hline
\multirow{2}{*}{} & 
\multicolumn{2}{c||}{\textbf{CPU w/}} & \multicolumn{3}{c||}{\textbf{GPU w/ concurrent}} &
\multicolumn{3}{c|}{\textbf{GPU w/ concurrent}} \\
{} & \multicolumn{2}{c||}{\textbf{Priority Queue}} & 
\multicolumn{3}{c||}{\textbf{ heap}} & 
\multicolumn{3}{c|}{\textbf{ heap and GC.}}\\%Garbage Collection}}\\
\hline
\multirow{2}{*}{\textbf{Dataset}}
 & \textbf{Time}  & \textbf{\# Nodes } & \textbf{Time}  & \textbf{\# Nodes} & \multirow{2}{*}{\textbf{SpeedUp}} & \textbf{Time}  & \textbf{\# Nodes} & 
\multirow{2}{*}{\textbf{SpeedUp}} \\
 {} & {\textbf{(ms)}} & \textbf{Explored} & {\textbf{(ms)}} & \textbf{Explored} & {}  & {\textbf{(ms)}} & \textbf{Explored} & {}\\
 
\hline
\textbf{ks\_sc\_700\_18k} & 825.70 & 782802 & 670.57 & 813089 & 1.23 & 595.58 & 810717 &1.39 \\
\textbf{ks\_sc\_800\_18k} & 977.49 & 923255 & 757.06 & 955374 & 1.29 & 708.02 & 956514 & 1.38 \\
\textbf{ks\_sc\_200\_7k} & 202.40 & 243106 & 199.76 & 249373 & 1.01 & 205.27 & 249935 & 0.99\\
\textbf{ks\_asc\_750\_16k} & 757.17 & 709267 & 722.90 & 389249 & 1.05 & 566.17 & 445231 & 1.34 \\
\textbf{ks\_asc\_1300\_6k} & 5239.97 & 4934552 & 5118.03 & 2832737 & 1.02 & 4115.55 & 2404241 & 1.27 \\
\textbf{ks\_asc\_500\_7k} & 502.37 & 475402 & 549.10 & 296824 & 0.91 & 499.03 & 295951 & 1.01 \\
\textbf{ks\_esc\_900\_18k} & 1128.4 & 1080182 & 848.65 & 1123920 & 1.33 & 796.58 & 1124880 & 1.42\\
\textbf{ks\_esc\_1200\_13k} & 2013.06 & 1950260 & 2066.64	& 2002002 & 0.97 & 1357.75	& 1770747 & 1.48\\
\textbf{ks\_esc\_400\_8k} & 355.25 & 399185 & 346.53 & 418925 & 1.03 & 348.52 & 421504 & 1.02\\
\textbf{ks\_ss\_100\_18k} & 42.38	& 54278 &	3.55 &	55 &	3.48 &	11.92 &	55 &	12.19\\
\textbf{ks\_ss\_1250\_12k} & 20.27 & 23886 & 4.30 & 94 & 4.12 & 4.72 & 94 & 4.92\\
\textbf{ks\_ss\_1300\_14k} & 25.02 & 25305 & 19.64 & 9452 & 1.27 & 18.01 &	8528 & 1.39\\
\bottomrule
\end{tabular}
}
\end{table}
In \cite{martello+MS99}, S. Martello et al. defined and tested with several types of instances of knapsack problems. Using the same data generation tool, we generated 12 knapsack datasets to demonstrate the potential of our concurrent heap. We describe the properties of these datasets in Table \ref{tab:knapsackData}.

%TODO: update average speedup
We compared the running time in milliseconds and the number of explored nodes between sequential and GPU knapsack in Table \ref{tab:knapsackResult}. For all the datasets we tested we obtain an average overall speedup of 2.31X for GPU knapsack and 2.48X for GPU knapsack with garbage collection.

We find that our GPU knapsack algorithms with concurrent heap performs particularly well with Subset-sum(ss) instances, with a maximum speedup up to 12.19X compared to sequential version. Also, the number of nodes explored for GPU knapsack is significantly smaller than that of sequential knapsack. Because of the greedy property of the branch and bound algorithm, it does not guarantee the path its exploring will lead to a global optimal solution, it is possible, especially with the Subset-sum instances where the benefit of an item is equal to the weight of it. On the other hand, parallelizing the branch and bound algorithm with our concurrent heap allows it to solve for a large amount of potential solutions that are prioritized by their current benefit simultaneously, which can lead to a faster convergence of the global optimal solution.

Theoretically, parallelizing node exploration in branch and bound knapsack problem may cause an exponential growth in the queue size since it performs parallel explorations for nodes in a binary decision tree. However, according to our experiments, we find that the GPU knapsack sometimes results in less node exploration because while nodes with higher benefit are explored earlier than other nodes, there are chances where the current max benefit converges fast enough so that the nodes with lower priority quickly become invalid for exploration since it's bound become smaller than the max benefit, which leads to a reduction in exploring time and correspondingly an increase of overall performance.

\section{Related Work}

\noindent {\bf \emph{CPU Parallel Heap Algorithms}}
The most
popular CPU approach \cite{Nageshwara+:TC88,Rassul:IPDPS90,Hunt+:IPL96,Nir+:PODC99,Sushil+:IPDPS95} to gain parallelism for parallel heap is by supporting concurrent insert or/and delete operations.  Rao and Kumar\cite{Nageshwara+:TC88} proposed a scheme that used multiple mutual exclusion locks to protect each node in the heap. They also proceeded insertions in the top-down manner to avoid deadlock while insertions for the conventional heap follow a bottom-up manner. 
Rassul\cite{Rassul:IPDPS90} proposed LR-algorithm which was an extension to Rao and Kumar's method that scatter the accesses of different operations to reduce the contention.
Hunt and others\cite{Hunt+:IPL96} present a lock-based heap algorithm that supports insertion and deletion in an opposite directions.
Deo and Prasad\cite{Deo+:JS92} increased the
node capacity. However, their algorithm does not support concurrent insertions/deletions. 

All these parallel heap algorithms on CPUs cannot be applied to GPUs directly because of the unique SIMT execution model employed by moder GPUs. For parallel algorithms on GPUs, the optimization for \textit{thread divergence}, \textit{memory coalescing} and \textit{synchronization} need to be taken into account.

\noindent {\bf \emph{GPU Parallel Heap Algorithms}} Parallel Heap algorithms are less studied on GPUs. He \cite{He+:HiPC12} introduced a parallel heap algorithms for many-fore architectures like GPUs. Their algorithm exploited the parallelism of the parallel heap by increasing the node capacity, like the idea in \cite{Deo+:JS92}, and pipelining the insert and delete operations. However, their approach did not exploit the parallelism for concurrent operations at the same level of the heap. Also they need a global barrier to synchronize all threads after the insert or delete updates at every levels which brought extra heavy overhead.

\section{Conclusion}

This work proposes a concurrent heap implementation that is friendly to many-core accelerators. We develop a generalized heap and support both intra-node and inter-node parallelism. We also prove that our two heap implementations are linearizable. We evaluate our concurrent heap thoroughly and show a maximum 19.49X speedup compared to the sequential CPU implementation and 2.11X speedup compared with the existing GPU implementation \cite{He+:HiPC12}.

%% Acknowledgments
\begin{acks}                            %% acks environment is optional
                                        %% contents suppressed with 'anonymous'
  %% Commands \grantsponsor{<sponsorID>}{<name>}{<url>} and
  %% \grantnum[<url>]{<sponsorID>}{<number>} should be used to
  %% acknowledge financial support and will be used by metadata
  %% extraction tools.
  This material is based upon work supported by the
  \grantsponsor{GS100000001}{National Science
    Foundation}{http://dx.doi.org/10.13039/100000001} under Grant
  No.~\grantnum{GS100000001}{nnnnnnn} and Grant
  No.~\grantnum{GS100000001}{mmmmmmm}.  Any opinions, findings, and
  conclusions or recommendations expressed in this material are those
  of the author and do not necessarily reflect the views of the
  National Science Foundation.
\end{acks}

%% Bibliography
%% \bibliographystyle{abbrv}
%% \bibliography{heapref}


%% Appendix
%\appendix
%\section{Appendix}

%Text of appendix \ldots

\end{document}